\documentclass[11pt]{article}

\usepackage[a4paper,margin=1in]{geometry}
\usepackage[T1]{fontenc}
\usepackage[utf8]{inputenc}
\usepackage{lmodern}
\usepackage{amsmath,amssymb,amsfonts,amsthm,mathtools}
\usepackage{bm}
\usepackage{graphicx}
\usepackage{booktabs}
\usepackage{multirow}
\usepackage{array}
\usepackage{float}
\usepackage{enumitem}

\usepackage[ruled,vlined,linesnumbered]{algorithm2e}

\usepackage[hidelinks]{hyperref}
\usepackage[nameinlink,capitalize]{cleveref}

\SetKwComment{Comment}{/* }{ */}

\newtheorem{theorem}{Theorem}[section]
\newtheorem{proposition}[theorem]{Proposition}
\newtheorem{lemma}[theorem]{Lemma}
\newtheorem{remark}[theorem]{Remark}

\newtheorem{corollary}[theorem]{Corollary}
\numberwithin{equation}{section}

\newcommand{\diag}{\operatorname{diag}}

\newcommand{\argmax}{\operatorname*{arg\,max}}

\newcommand{\Green}{\mathcal{G}}

\title{\textbf{Forward--Backward Green Cosine Geometry 
for Directed Community Detection and Overlap Expansion}}

\author{
Duy Hieu DO\\
Institute of Mathematics, Vietnam Academy of Science and Technology
}
\date{}

\begin{document}

\maketitle

\begin{center}
\texttt{ddhieu@math.ac.vn}
\end{center}

\begin{abstract}
Community detection in directed graphs is challenging because edge asymmetry
induces non-reversible diffusion, direction-dependent accessibility, and
different source and target roles. The difficulty is even greater in
overlapping community detection, where a vertex may belong to more than one
group.
This paper develops a Green-based cosine geometry for directed community
detection and for expanding a disjoint partition into an overlapping cover.
The starting point is the observation that hitting-time information is natural
for directed graphs, but raw hitting-time vectors are not well suited for
cosine comparison. Indeed, they contain a
source-independent stationary baseline, while cosine similarity is not
translation-invariant. We therefore replace raw hitting-time profiles by centered Green profiles of
the directed random walk. For constructing vertex coordinates, we use the
diffusive part of the truncated Green profile, excluding the time-zero
self-spike. To account for edge asymmetry, we combine the diffusive Green
profile of the original walk with the corresponding profile on the
edge-reversed graph, obtaining a forward--backward Green coordinate for each
vertex.

The proposed framework leads to two concrete algorithms based on the same
forward--backward Green cosine geometry. The first is
Di-Green-FB-cosine-KMeans, a disjoint directed community detection algorithm
that clusters vertices in the proposed Green cosine space. The second is
Di-Green-FB-Cosine Overlap, an overlap expansion algorithm that converts an
initial disjoint partition into an overlapping cover by a community-adaptive
cosine rule. The framework is modular: the initial partition can in principle
be supplied by any disjoint community detection algorithm, but in the main
pipeline of this paper it is produced by our own
Di-Green-FB-cosine-KMeans method.  Experiments on synthetic directed benchmarks show that the proposed Green
forward--backward geometry substantially improves over raw hitting-time cosine
variants and gives competitive performance against directed spectral and
flow-based baselines for disjoint detection. Additional real-network
experiments, evaluated only by directed modularity as an internal
structural-quality measure, indicate that the same geometry can produce
coherent directed partitions on diverse empirical networks. Finally, synthetic overlap experiments show that the proposed geometry can
recover additional memberships effectively, especially in moderate and weakly
separated directed networks.
\end{abstract}


\section{Introduction}

Community detection is a central topic in network science. Its purpose is to
partition the vertices of a graph into groups that reflect some meaningful
internal organization, such as dense internal connectivity, coherent
transport behavior, or similar functional roles. In undirected graphs, the
subject has been extensively studied, and a broad spectrum of methods has been
developed, including modularity optimization, spectral clustering,
probabilistic block models, and random-walk-based approaches
\cite{Fortunato2010}. In directed graphs, however, the problem is more
subtle. Because the edge set is asymmetric, the graph may exhibit
non-reversible transport, direction-sensitive accessibility, and distinct
sending and receiving roles. Consequently, methods that are geometrically
natural in the undirected setting may become incomplete or even misleading in
the directed case \cite{MalliarosVazirgiannis2013,Chung2005}.

The situation becomes even more challenging when communities are allowed to
overlap. In many real systems, a vertex may simultaneously belong to several
communities. This phenomenon is common in social, biological, information,
and technological networks, and has motivated a substantial literature on
overlapping community detection \cite{Peel2017,Jebabli2018,Harenberg2014}. A
standard approach is to first compute a disjoint partition and then enlarge it
by assigning some vertices to additional communities
\cite{Ponomarenko2021,Chen2010}. The success of such a two-step framework
depends heavily on the criterion used to compare a vertex with a community.
If that criterion is too local, then it may fail to capture the genuinely
global transport structure of a directed graph.

A natural response is to adopt a random-walk viewpoint. Random walks have long
played an important role in community detection because they reflect how flow
or information propagates through the network: a walk tends to remain for a
relatively long time inside a coherent region before escaping it
\cite{PonsLatapy2006,RosvallBergstrom2008}. This viewpoint has also motivated
a line of our recent work on random-walk-based community detection, including
an extended Walktrap method for large networks, random-walk refinements of
Louvain-type algorithms, and hitting-time-based methods for directed graphs
\cite{DoPhan2022ExtendedWalktrap,DoPhan2025LouvainRW,DoNguyenPhan2025DFLouvainRW,DangDoPhan2023}. 
In directed graphs, this perspective is especially natural, since local degree
information alone often fails to describe the large-scale organization induced
by asymmetric transport \cite{MalliarosVazirgiannis2013}. Once one takes this
viewpoint seriously, a vertex should no longer be represented merely by its
adjacency pattern, but by the way it reaches the rest of the graph under the
directed random walk.

From this random-walk perspective, hitting times and related accessibility
quantities are natural candidates for vertex representation. If two vertices
belong to the same community, then one expects them to ``see'' the rest of the
graph in similar ways, and hence their global accessibility profiles should be
similar. This viewpoint is closely related to hitting-time-based approaches
for directed community detection, such as the framework of
Dang, Do, and Phan~\cite{DangDoPhan2023}, where vertices are represented
through stationary-weighted hitting-time information. Such methods support the
idea that hitting times encode meaningful global accessibility structure in
directed graphs.

The present paper asks a different geometric question. Suppose one wants to
compare hitting-time-type accessibility profiles by cosine similarity, so that
community similarity is interpreted as alignment of profiles rather than as a
distance between them. This choice is natural for profile-based representations:
cosine similarity compares the directions of two vectors and is therefore
insensitive to their overall scale. It is especially useful when vertices are
represented by high-dimensional accessibility, diffusion, or embedding
profiles, where the norm may reflect activity level, degree, or total mass,
whereas the direction reflects the relative pattern of interaction. This
principle is classical in spherical clustering \cite{DhillonModha2001,Hornik2012}
and is also closely related to modern graph representation methods, where
vertices are mapped to vector spaces and then compared by inner-product or
cosine-type similarities \cite{Perozzi2014,Tang2015,GroverLeskovec2016}.
Cosine-type indices have also appeared in network similarity and link
prediction, for example through the Salton index \cite{LuZhou2011}, and recent
embedding-based approaches to directed community detection also use cosine
similarity between node representations \cite{YuJiaoDehmerEmmertStreib2024}.
Cosine similarity has also been used directly in overlap-expansion frameworks;
in particular, our previous directed cosine-based overlap method assigns
additional memberships by comparing vertices and communities through cosine
scores \cite{DoPhan2025OverlapCosine}.

The present paper continues this cosine-based line, but changes the object of
comparison. Instead of applying cosine similarity to the earlier directed
cosine representation, we construct stationary-centered forward--backward
Green diffusion profiles and apply cosine similarity to these corrected
accessibility coordinates. Thus, the use of cosine similarity is not
arbitrary: it provides a scale-free way to compare the alignment of vertex
profiles. However, it also raises a crucial question: which accessibility
profiles are appropriate objects for cosine comparison?

This question is important because cosine similarity, although powerful for
comparing profile alignment, is sensitive to translations of the feature
space. Raw hitting-time vectors are therefore not directly compatible with
cosine comparison. The reason is that they contain a common stationary
baseline; equivalently, they differ from a centered accessibility profile by a
translation that is independent of the source vertex. Since cosine similarity
is not translation-invariant, applying cosine directly to raw hitting-time
vectors can distort the intended geometry. Thus, the difficulty is not that
hitting-time information is irrelevant, but rather that raw hitting-time
coordinates are not the correct objects on which cosine should act.

This observation leads naturally to the Green operator of the directed random
walk. Green profiles may be viewed as the centered counterparts of
hitting-time-type information, obtained after removing the stationary
background. For a source vertex \(i\) and a target vertex \(k\), the Green
entry \(\Green(i,k)\) measures whether \(k\) is reached from \(i\) more
easily or less easily than under the stationary regime
\cite{KemenySnell1976,Hunter1982,LiZhang2012}. Hence the row
\(\Green(i,\cdot)\) records the full \emph{relative accessibility profile} of
vertex \(i\). If two vertices belong to the same community, then they should
play similar roles in the global transport structure of the graph, and thus
their Green profiles should be similar. In this sense, Green geometry is not
an arbitrary algebraic construction, but the natural corrected geometry that
emerges when one tries to combine hitting-time intuition with cosine
comparison.

Yet in a directed graph, a one-sided Green profile is still not always
sufficient. The row \(\Green(i,\cdot)\) captures how the walk starting from
\(i\) reaches the rest of the graph, namely the outgoing or
source-to-target aspect of directed diffusion. But a directed vertex also has
an incoming role: it may be reached from the rest of the graph in a way that
is structurally different from another vertex with a similar outgoing
profile. Thus, two vertices may have similar forward accessibility while
belonging to quite different incoming environments. If one seeks a faithful
geometry for directed community detection, both aspects should be taken into
account.

This leads to the main methodological idea of the paper. For each vertex, we
combine a forward Green profile and a backward Green profile into a single
forward--backward Green coordinate. The forward profile describes how the
vertex reaches the rest of the network, whereas the backward profile describes
how the rest of the network reaches that vertex through reversed diffusion.
Thus, the coordinate captures both outgoing and incoming diffusion roles.

We use this geometry in a two-stage framework. First, an initial disjoint
partition is constructed, either by an existing non-overlapping community
detection method or by the proposed Di-Green-FB-cosine-KMeans algorithm.
Second, a community-adaptive cosine rule expands this partition into an
overlapping cover by assigning a vertex to an additional community when its
forward--backward Green profile is sufficiently aligned with that community.
Thus, both stages are governed by the same principle: community membership is
interpreted as alignment in forward--backward Green diffusion space.

Accordingly, the paper develops two concrete algorithms. The first algorithm,
Di-Green-FB-cosine-KMeans, is designed for disjoint directed community
detection. It constructs forward--backward Green diffusion coordinates and
applies spherical \(K\)-means in the resulting cosine geometry. The second
algorithm, Di-Green-FB-Cosine Overlap, is designed for overlapping directed
community detection. It starts from an initial disjoint partition and expands
it into an overlapping cover using community-adaptive Green cosine thresholds.
In this way, the disjoint and overlapping tasks are treated within a single
geometric framework.

The paper also gives basic theoretical support for this construction, including
centering identities for the Green operator, a baseline-distortion result for
raw hitting-time cosine, a finite-time truncation bound, a kernel property for
the forward--backward cosine, and deterministic margin conditions for the
idealized cosine assignment rules.

\subsection{Main contributions}

The main contributions of this paper are as follows.

\begin{itemize}
    \item \textbf{Baseline distortion of raw hitting-time cosine.}
    We identify a geometric obstruction in applying cosine similarity directly
    to raw hitting-time profiles. Although hitting times contain useful
    accessibility information for directed graphs, their rows include a common
    source-independent stationary baseline. Since cosine similarity is not
    translation-invariant, this baseline may dominate the comparison and distort
    the intended community geometry.

    \item \textbf{Centered Green diffusion geometry.}
    We replace raw hitting-time profiles by centered Green diffusion profiles.
    The Green operator removes the stationary background and represents the
    source-dependent deviation of a directed random walk from stationarity. We
    connect this representation to hitting-time deviations, Poisson-type
    identities on the stationary-free subspace, and directed spectral geometry.

    \item \textbf{Forward--backward Green cosine representation.}
    To capture the asymmetric nature of directed graphs, we combine the Green
    profile of the original random walk with the corresponding profile on the
    edge-reversed graph. The resulting forward--backward coordinate encodes both
    outgoing and incoming diffusion roles of a vertex. We show that the induced
    forward--backward cosine similarity is a positive semidefinite kernel on the
    normalized Green coordinates.

    \item \textbf{Algorithms for disjoint and overlapping community detection.}
    Based on the same Green cosine geometry, we propose two concrete
    algorithms. The first is Di-Green-FB-cosine-KMeans, a disjoint directed
    clustering algorithm that constructs forward--backward Green diffusion
    coordinates and applies spherical \(K\)-means in the resulting cosine
    space. The second is Di-Green-FB-Cosine Overlap, a community-adaptive
    expansion algorithm that converts an initial disjoint partition into an
    overlapping cover. Thus, the paper provides a unified geometric framework
    for both non-overlapping and overlapping community detection in directed
    networks.
    \item \textbf{Experimental evaluation.}
    We evaluate the proposed framework on synthetic directed benchmarks,
    overlapping directed benchmarks, and real directed networks. The experiments
    compare the proposed method with raw hitting-time cosine variants, directed
    spectral baselines, flow-based references, and overlap-expansion baselines.
    The results show that centering through Green profiles and combining
    forward and backward diffusion information improve the robustness of cosine
    geometry for directed community detection.
\end{itemize}

\subsection{Organization of the paper}

The remainder of the paper is organized as follows. In
\cref{sec:related} we review related work. In
\cref{sec:green-geometry} we develop the forward--backward Green geometry and
its basic theoretical guarantees. This section starts from teleported directed
random walks, stationary centering, Green operators, and the hitting-time
baseline problem, and then introduces the resulting forward--backward cosine
geometry. In \cref{sec:green-fb-disjoint} we present the first algorithm,
Di-Green-FB-cosine-KMeans, for disjoint directed community detection. In
\cref{sec:green-fb-overlap} we present the second algorithm,
Di-Green-FB-Cosine Overlap, for expanding a disjoint partition into an
overlapping cover, and we discuss its computational complexity. In
\cref{sec:experiments} we evaluate the proposed framework on synthetic
directed benchmarks, overlapping directed benchmarks, and real directed
networks. The paper ends with concluding remarks and future directions.

\section{Related Work}
\label{sec:related}

Community detection in directed graphs has been studied from several
viewpoints, including modularity optimization, probabilistic block models,
directed spectral clustering, and random-walk-based approaches
\cite{Fortunato2010,MalliarosVazirgiannis2013}. Since the present work is
motivated by a random-walk and hitting-time perspective, the most relevant
prior work lies in the random-walk, Markov-chain, and directed spectral
directions.

A broad class of methods uses random walks and flow-based quantities to infer
community structure. These methods rely on the fact that a random walk tends
to spend relatively long times inside coherent regions of the graph, thereby
revealing community boundaries through flow persistence, multi-step
accessibility, or first-passage behavior \cite{PonsLatapy2006,RosvallBergstrom2008}.
In directed graphs, this viewpoint is particularly natural because asymmetric
transport may create large-scale structures that are invisible to purely local
edge statistics \cite{MalliarosVazirgiannis2013}. The present paper is close
in spirit to this tradition, but differs in that it seeks a vertex geometry
built from global relative accessibility profiles rather than from local flow
or short-time transition information alone.

The random-walk viewpoint also naturally leads to hitting times, the
fundamental matrix, and related Green-function-type quantities. These objects
form a classical part of finite Markov chain theory
\cite{KemenySnell1976,Hunter1982} and have also been used in
hitting-time-based methods for directed community detection
\cite{DangDoPhan2023}. In particular, they provide global
descriptions of accessibility and first-passage structure, and thereby offer
a natural source of vertex representations for community analysis. Our work
draws directly on this perspective. At the same time, a central message of the
present paper is that if one wants to compare hitting-time-type profiles by
cosine similarity, then one should not use raw hitting-time vectors directly.
Instead, one should first remove the stationary background, which leads
naturally to Green coordinates.

Cosine similarity has also played an important role in clustering, network
analysis, graph representation learning, and overlap expansion. In
high-dimensional data analysis, spherical \(K\)-means uses cosine or angular
similarity as the basic clustering criterion
\cite{DhillonModha2001,Hornik2012}. In network science, cosine-type
normalizations appear in classical node-similarity indices such as the Salton
index \cite{LuZhou2011}. More broadly, graph representation learning methods
such as DeepWalk, LINE, and node2vec represent vertices as vectors whose
relative positions encode network neighborhoods or structural proximity
\cite{Perozzi2014,Tang2015,GroverLeskovec2016}. Community detection based on
graph embeddings often follows this general strategy: first construct a
vector representation of vertices, and then cluster or compare vertices in
the embedding space \cite{Tandon2021,YuJiaoDehmerEmmertStreib2024}.

The present paper is also related to our previous directed cosine-based
overlap expansion method \cite{DoPhan2025OverlapCosine}, where cosine
similarity was used to assign additional memberships after an initial
disjoint partition had been obtained. The key difference is the representation
on which cosine similarity acts. Instead of applying cosine similarity to
adjacency-type features, generic embeddings, the earlier directed cosine
representation, or raw hitting-time rows, we apply it to
stationary-centered forward--backward Green diffusion profiles. Thus, the
novelty is not merely the use of cosine similarity, but the construction of a
directed Green geometry on which cosine comparison becomes meaningful.

A second closely related line of work concerns directed spectral geometry. In
directed graphs, Laplacian-like operators are generally non-symmetric, which
requires singular-value-based constructions or carefully normalized operators.
Directed Laplacian and Diplacian frameworks provide principled ways to extract
meaningful low-dimensional structure from non-reversible random walks
\cite{Chung2005,LiZhang2012}. In particular, the Diplacian offers a spectral
counterpart to Green and fundamental-matrix-based viewpoints
\cite{LiZhang2012}. Our approach is not a purely spectral method, but it is
conceptually related to this literature because it also seeks a directed
vertex geometry governed by random-walk structure rather than by local edge
density alone.

Overlapping community detection has been extensively studied in undirected
networks \cite{Peel2017,Jebabli2018,Harenberg2014}. A common strategy is to
first compute a disjoint partition and then enlarge it using a local,
semilocal, or similarity-based criterion \cite{Ponomarenko2021,Chen2010}. Our
method belongs to this two-step family and is directly connected to our
previous directed cosine-based overlap expansion method
\cite{DoPhan2025OverlapCosine}. However, the present work changes the
geometric representation used in the expansion step. Rather than relying on
local adjacency counts or the earlier directed cosine representation, it uses
stationary-centered forward--backward Green diffusion profiles.

This distinction is important because the proposed Green representation is
used not only for overlap expansion, but also for constructing the initial
disjoint partition. Thus, the same geometry governs both stages of the
framework: first to construct a disjoint directed partition, and then to
assign overlapping memberships by a community-adaptive cosine rule.

The present work combines these viewpoints by using Green-function-based
random-walk profiles as a directed vertex geometry. Its main distinction is
that both the disjoint clustering step and the overlap assignment step are
performed in the same forward--backward Green cosine space, rather than using
local edge counts, density gains, or a separate overlap heuristic.

\section{Forward--Backward Green Geometry and Basic Guarantees}
\label{sec:green-geometry}

This section develops the geometric foundation of the proposed method. The
main idea is to represent each vertex by a centered random-walk accessibility
profile and to compare such profiles by cosine similarity. We first define
forward and backward teleported random walks. We then introduce the centered
Green operator, explain its relation with hitting times, and show why raw
hitting-time rows are not suitable for direct cosine comparison. Finally, we
combine forward and backward Green profiles into a single directed cosine
geometry and record several basic theoretical guarantees.

\subsection{Teleported forward and backward random walks}

Let \(G=(V,E,A)\) be a finite directed weighted graph with vertex set \(V\),
\(|V|=n\), and nonnegative weighted adjacency matrix
\[
A=(A_{ij})_{1\le i,j\le n}\in\mathbb R^{n\times n}.
\]
Here \(A_{ij}>0\) means that there is a directed edge from \(i\) to \(j\).
Weighted edges are allowed, since the random-walk construction and the Green
operator apply without change to nonnegative weighted graphs.

We identify the vertex set \(V\) with \(\{1,\ldots,n\}\). Throughout the
paper, \(\mathbf 1\in\mathbb R^n\) denotes the all-ones column vector, and
\(I\) denotes the identity matrix of the appropriate size. For a matrix \(M\),
we write \(M_{i\bullet}\) for its \(i\)-th row. The notation
\(\mathbf 1\{\cdot\}\) denotes the indicator of an event or condition.

For a nonnegative matrix \(B\), define its row-normalized transition matrix
\(\widetilde P(B)\) by
\[
\widetilde P(B)_{ij}
=
\begin{cases}
\dfrac{B_{ij}}{\sum_k B_{ik}}, & \sum_k B_{ik}>0,\\[6pt]
\dfrac{1}{n}, & \sum_k B_{ik}=0.
\end{cases}
\]
Thus dangling rows are replaced by the uniform distribution. Given a
teleportation parameter \(0<\alpha<1\), as in PageRank-type random-walk
regularization \cite{Page1999}, define the teleported transition matrix
\[
P_\alpha(B)
=
\alpha \widetilde P(B)
+
(1-\alpha)\frac{1}{n}\mathbf 1\mathbf 1^\top.
\]
Teleportation has two roles. First, it guarantees that the Markov chain is
irreducible and aperiodic. Second, it makes the stationary distribution and
Green operator well defined even when the original directed graph is not
strongly connected. In applications one usually takes \(\alpha\) close to
one, so that the walk still follows the directed edges most of the time while
remaining globally regularized.

In a directed graph, outgoing and incoming diffusion may carry different
structural information. We therefore use two teleported walks. The forward walk is
\begin{equation}
P_\alpha^{+}=P_\alpha(A),
\label{eq:forward-teleport}
\end{equation}
and the backward walk is the teleported walk on the edge-reversed graph,
\begin{equation}
P_\alpha^{-}=P_\alpha(A^\top).
\label{eq:backward-teleport}
\end{equation}
The forward walk describes how a vertex reaches the rest of the graph in the
original edge orientation. The backward walk describes diffusion on the graph
with all edges reversed, and is used to encode incoming accessibility patterns
of vertices in the original graph.

We emphasize that \(P_\alpha^{-}\) is not the Markov-chain time reversal of
\(P_\alpha^{+}\). The time reversal of a Markov chain with transition matrix
\(P\) and stationary distribution \(\pi\) depends on \(\pi\) and has transition
probabilities
\[
P^{\ast}_{ij}
=
\frac{\pi_j P_{ji}}{\pi_i}.
\]
In contrast, \(P_\alpha^{-}=P_\alpha(A^\top)\) is the ordinary random walk on
the edge-reversed graph. Its role in this paper is structural: it provides a
second diffusion view that captures how a vertex is accessed through incoming
edges.

We denote the stationary distributions of \(P_\alpha^{+}\) and
\(P_\alpha^{-}\) by \(\pi^{+}\) and \(\pi^{-}\), respectively.

\subsection{Stationary background and centered Green operator}

We now describe the Green operator for a generic teleported transition matrix
\(P\). Let \(\phi\) be the unique stationary distribution of \(P\):
\[
\phi^\top P=\phi^\top,
\qquad
\phi_i>0,
\qquad
\sum_i\phi_i=1.
\]
Let
\[
\Phi=\diag(\phi_1,\ldots,\phi_n),
\qquad
\Pi=\mathbf 1\phi^\top.
\]
The rank-one matrix \(\Pi\) is the stationary background of the chain. Every
row of \(\Pi\) is equal to \(\phi^\top\), and
\[
P^t\longrightarrow \Pi
\qquad\text{as }t\to\infty.
\]
Thus, after many steps, the walk no longer remembers its starting vertex. For
community detection, however, the relevant information is precisely the
source-dependent deviation from this limiting background. This is why the
central object in the paper is the centered Green operator.

The fundamental matrix is
\begin{equation}
Z(P)
=
\left(I-P+\Pi\right)^{-1}.
\label{eq:fundamental-new}
\end{equation}
Equivalently,
\[
Z(P)=\sum_{t=0}^{\infty}(P-\Pi)^t.
\]
Since \((P-\Pi)^t=P^t-\Pi\) for \(t\geq 1\), the centered Green operator is
\begin{equation}
\mathcal G(P)
=
Z(P)-\Pi
=
\sum_{t=0}^{\infty}(P^t-\Pi).
\label{eq:green-series}
\end{equation}
The entry \(\mathcal G(P)_{ik}\) accumulates, over all time scales, the excess
or deficit of the probability of being at \(k\) when the walk starts from
\(i\), relative to the stationary regime.

\begin{proposition}[Poisson identities for the Green operator]
\label{prop:green-poisson}
The Green operator satisfies
\[
\mathcal G(P)\mathbf 1=0,
\qquad
\phi^\top \mathcal G(P)=0,
\]
and
\[
(I-P)\mathcal G(P)=I-\Pi,
\qquad
\mathcal G(P)(I-P)=I-\Pi.
\]
Thus, \(\mathcal G(P)\) plays the role of the inverse of \(I-P\) on the
stationary-free subspace. More precisely, it inverts \(I-P\) after the
rank-one stationary component represented by \(\Pi\) has been removed.
\end{proposition}

\begin{proof}
Using
\[
\mathcal G(P)=\sum_{t=0}^{\infty}(P^t-\Pi),
\]
we have
\[
(P^t-\Pi)\mathbf 1=P^t\mathbf 1-\Pi\mathbf 1=\mathbf 1-\mathbf 1=0.
\]
Hence \(\mathcal G(P)\mathbf 1=0\). Similarly, since
\(\phi^\top P^t=\phi^\top\) and \(\phi^\top\Pi=\phi^\top\), we obtain
\[
\phi^\top(P^t-\Pi)=0
\]
for every \(t\), and therefore \(\phi^\top\mathcal G(P)=0\).

For \(T\geq 0\),
\[
(I-P)\sum_{t=0}^T(P^t-\Pi)
=
\sum_{t=0}^T(P^t-P^{t+1})
=
I-P^{T+1},
\]
because \(P\Pi=\Pi\). Letting \(T\to\infty\) gives
\[
(I-P)\mathcal G(P)=I-\Pi.
\]
The identity on the right is proved similarly, using \(\Pi P=\Pi\):
\[
\sum_{t=0}^T(P^t-\Pi)(I-P)
=
I-P^{T+1}
\longrightarrow I-\Pi.
\]
\end{proof}

This proposition gives the precise sense in which the Green operator is a
centered inverse of the random-walk operator. It removes the stationary
direction and keeps only the source-dependent accessibility structure.

\subsection{Hitting-time baseline and cosine distortion}

The Green operator is closely connected with hitting times. Let \(H(i,k)\)
denote the expected hitting time from vertex \(i\) to vertex \(k\), with the
convention
\[
    H(k,k)=0.
\]
Define
\[
H(\phi,k)=\sum_{x\in V}\phi_x H(x,k).
\]
With this convention, the standard fundamental-matrix identity gives
\[
    H(i,k)=\frac{Z(P)_{kk}-Z(P)_{ik}}{\phi_k}.
\]
Equivalently,
\begin{equation}
\mathcal G(P)_{ik}
=
\phi_k\bigl(H(\phi,k)-H(i,k)\bigr).
\label{eq:green-hitting-new}
\end{equation}
Thus \(\mathcal G(P)_{ik}\) is positive when \(k\) is reached from \(i\)
earlier than under the stationary baseline, and negative when \(k\) is reached
later than under that baseline.

This identity also explains why raw hitting-time rows are not the right
objects for cosine comparison. For each source vertex \(i\), define
\[
H_i=(H(i,1),\ldots,H(i,n)),
\qquad
h=(H(\phi,1),\ldots,H(\phi,n)).
\]
Then \eqref{eq:green-hitting-new} implies
\begin{equation}
H_i
=
h-\mathcal G(P)_{i\bullet}\Phi^{-1}.
\label{eq:hitting-baseline-decomposition}
\end{equation}
Therefore every raw hitting-time row contains the same source-independent
baseline \(h\). The source-dependent information is carried by the centered
Green profile, after stationary reweighting.

\begin{proposition}[Hitting-time baseline and cosine distortion]
\label{prop:hitting-cosine-baseline}
Raw hitting-time rows have the decomposition
\[
H_i=h-\mathcal G(P)_{i\bullet}\Phi^{-1},
\]
where \(h\) is independent of the source vertex \(i\). Since cosine similarity
is not translation-invariant, comparing the raw rows \(H_i\) by cosine may
mainly measure their common baseline rather than their source-dependent
accessibility profiles.

A simple limiting calculation illustrates this obstruction. If \(h\neq 0\)
and \(a,b\in\mathbb R^n\), then
\[
\lim_{\gamma\to\infty}
\frac{\langle \gamma h+a,\gamma h+b\rangle}
{\|\gamma h+a\|_2\|\gamma h+b\|_2}
=
1.
\]
Thus, when a large common baseline is present, cosine similarity can become
insensitive to the source-dependent deviations \(a\) and \(b\).
\end{proposition}

\begin{proof}
The decomposition follows directly from
\[
\mathcal G(P)_{ik}
=
\phi_k\bigl(H(\phi,k)-H(i,k)\bigr),
\]
which is equivalent to
\[
H(i,k)=H(\phi,k)-\frac{\mathcal G(P)_{ik}}{\phi_k}.
\]
Writing this identity for all \(k\) gives
\[
H_i=h-\mathcal G(P)_{i\bullet}\Phi^{-1}.
\]

It remains to justify the baseline distortion statement. Dividing numerator
and denominator by \(\gamma^2\), we obtain
\[
\frac{\langle \gamma h+a,\gamma h+b\rangle}{\gamma^2}
=
\|h\|_2^2
+
\frac{\langle h,a\rangle+\langle h,b\rangle}{\gamma}
+
\frac{\langle a,b\rangle}{\gamma^2},
\]
while
\[
\frac{\|\gamma h+a\|_2}{\gamma}\to \|h\|_2,
\qquad
\frac{\|\gamma h+b\|_2}{\gamma}\to \|h\|_2.
\]
Therefore the cosine tends to \(1\).
\end{proof}

This proposition does not say that hitting times are uninformative. Rather, it
says that if hitting-time information is to be compared by cosine similarity,
one should first remove the common stationary baseline. The Green operator is
precisely the centered object that performs this correction.

\begin{remark}[Connection with directed spectral geometry]
The Green operator also has a spectral interpretation. Define
\[
\Gamma=\Phi^{1/2}(I-P)\Phi^{-1/2}.
\]
This operator is a normalized directed random-walk operator whose null
direction is \(\Phi^{1/2}\mathbf 1\). The Green operator is the group inverse
of \(I-P\) on the stationary-free subspace, and its normalized version
\[
    \Phi^{1/2}\mathcal G(P)\Phi^{-1/2}
\]
is the corresponding inverse of \(\Gamma\) on the subspace orthogonal to
\(\Phi^{1/2}\mathbf 1\).

Thus the Green viewpoint and the directed spectral viewpoint are closely
connected: both describe the same centered random-walk geometry, one from a
probabilistic accessibility perspective and the other from an
operator-theoretic perspective. A singular-vector representation of this
normalized inverse gives the corresponding directed spectral form.
\end{remark}

\subsection{Forward--backward Green cosine coordinates}
\label{sec:green-fb-coordinates}

The row \(\mathcal G(P)_{i\bullet}\) records how the walk starting from \(i\)
sees the rest of the graph after the stationary background has been removed.
Vertices in the same directed community should therefore have similar
relative accessibility profiles. This gives the basic geometric principle:
vertices are close when their centered diffusion profiles are aligned.

In a directed graph, however, one profile is not enough. The forward Green
profile captures the outgoing or source-to-target role of a vertex. A vertex
also has an incoming role: it may be reached from the rest of the graph in a
way that is not determined by its outgoing behavior. Hence a directed
community geometry should combine both forward and backward accessibility.

For constructing vertex coordinates, we do not use the full Green row
directly. Instead, we use its post-transition, or diffusive, part. For a
transition matrix \(P\) with stationary distribution \(\pi\), define
\begin{equation}
\mathcal G_{T}^{\mathrm{diff}}(P)
=
\sum_{t=1}^{T}
\left(P^t-\mathbf 1\pi^\top\right).
\label{eq:green-truncated-diffusive}
\end{equation}
The sum starts from \(t=1\), rather than \(t=0\). Recall that the full centered
Green operator is
\[
    \mathcal G(P)
    =
    \sum_{t=0}^{\infty}
    \left(P^t-\mathbf 1\pi^\top\right).
\]
Therefore its infinite diffusive part is
\[
    \mathcal G_{\geq 1}(P)
    =
    \sum_{t=1}^{\infty}
    \left(P^t-\mathbf 1\pi^\top\right)
    =
    \mathcal G(P)-\left(I-\mathbf 1\pi^\top\right).
\]
Thus the coordinate used by the algorithm is a truncated approximation of
\(\mathcal G_{\geq 1}(P)\), not of the full Green row
\(\mathcal G(P)\) itself.

This distinction is intentional. The removed term
\(I-\mathbf 1\pi^\top\) contains the time-zero self-spike, which records the
trivial fact that a walk starting from \(u\) is initially located at \(u\).
Since our goal is to compare how vertices diffuse through the graph after the
walk has moved, this self-spike is not used as part of the clustering
coordinate.

\begin{remark}[Green diffusion profiles versus reweighted hitting-time profiles]
The identity
\[
    H_i=h-\mathcal G(P)_{i\bullet}\Phi^{-1}
\]
shows that \(\mathcal G(P)\Phi^{-1}\) is the centered counterpart of the raw
hitting-time row. In the algorithm, however, we use the Green diffusion profile
\(\mathcal G_T^{\mathrm{diff}}(P)_{i\bullet}\) itself. This choice is
intentional. The Green row measures excess transition probability relative to
stationarity, whereas multiplication by \(\Phi^{-1}\) converts this quantity
back to a hitting-time scale and may strongly amplify vertices with small
stationary probability. Thus, the diagnostic experiment includes
\(\mathcal G(P)\Phi^{-1}\) to demonstrate the baseline correction, while the
proposed algorithm uses the unreweighted diffusive Green profile as a
probability-deviation coordinate.
\end{remark}

For each vertex \(u\), define the forward and backward diffusive Green
profiles by
\begin{equation}
x_u^{+}
=
\mathcal G_{T}^{\mathrm{diff}}(P_\alpha^{+})_{u\bullet},
\qquad
x_u^{-}
=
\mathcal G_{T}^{\mathrm{diff}}(P_\alpha^{-})_{u\bullet}.
\label{eq:forward-backward-green-profiles}
\end{equation}
The vector \(x_u^{+}\) describes the accumulated excess accessibility from
\(u\) to the rest of the graph after at least one transition, while
\(x_u^{-}\) gives the analogous profile on the edge-reversed graph.

We normalize the two profiles by
\begin{equation}
\widehat x_u^{+}
=
\frac{x_u^{+}}{\|x_u^{+}\|_2},
\qquad
\widehat x_u^{-}
=
\frac{x_u^{-}}{\|x_u^{-}\|_2}.
\label{eq:green-normalization}
\end{equation}
The theoretical discussion assumes nonzero norms. In numerical
implementations, if a one-sided norm is below a tolerance \(\tau>0\), we use
the regularized normalization
\[
\widehat x_u^{+}
=
\frac{x_u^{+}}{\max\{\|x_u^{+}\|_2,\tau\}},
\qquad
\widehat x_u^{-}
=
\frac{x_u^{-}}{\max\{\|x_u^{-}\|_2,\tau\}}.
\]

For \(\lambda\in[0,1]\), define the preliminary forward--backward coordinate
\[
\widetilde z_u^{(\lambda)}
=
\left(
\sqrt{\lambda}\,\widehat x_u^{+},
\sqrt{1-\lambda}\,\widehat x_u^{-}
\right).
\]
In the nondegenerate theoretical setting, where both one-sided profiles have
unit norm after normalization, we have
\(\|\widetilde z_u^{(\lambda)}\|_2=1\), and we set
\[
    z_u^{(\lambda)}=\widetilde z_u^{(\lambda)}.
\]
In numerical implementations, to keep the cosine geometry well defined even
in nearly degenerate cases, we apply a final normalization
\begin{equation}
z_u^{(\lambda)}
=
\frac{\widetilde z_u^{(\lambda)}}
{\max\{\|\widetilde z_u^{(\lambda)}\|_2,\tau\}}.
\label{eq:fb-green-coordinate}
\end{equation}
The square-root weights are chosen so that the inner product is a convex
combination of forward and backward cosine similarities. Define
\[
\cos^{+}(u,v)
=
\left\langle
\widehat x_u^{+},
\widehat x_v^{+}
\right\rangle,
\qquad
\cos^{-}(u,v)
=
\left\langle
\widehat x_u^{-},
\widehat x_v^{-}
\right\rangle.
\]
Then, in the nondegenerate case \(\|z_u^{(\lambda)}\|_2=1\), we have
\begin{align}
\left\langle z_u^{(\lambda)},z_v^{(\lambda)}\right\rangle
&=
\lambda
\left\langle
\widehat x_u^{+},
\widehat x_v^{+}
\right\rangle
+
(1-\lambda)
\left\langle
\widehat x_u^{-},
\widehat x_v^{-}
\right\rangle \notag\\
&=
\lambda \cos^{+}(u,v)
+
(1-\lambda)\cos^{-}(u,v).
\label{eq:fb-green-inner-product}
\end{align}
We define the forward--backward Green cosine similarity by
\begin{equation}
\operatorname{cos}_{\mathrm{FB}}(u,v)
=
\left\langle z_u^{(\lambda)},z_v^{(\lambda)}\right\rangle.
\label{eq:fb-green-cosine}
\end{equation}
In the experiments we use \(\lambda=1/2\), giving equal weight to outgoing
and incoming diffusion roles.

The resulting principle is the following: two vertices are likely to belong
to the same directed community when their forward accessibility profiles and
their backward accessibility profiles are both well aligned. This single
geometry will be used twice in the paper: first to construct an initial
disjoint partition, and then to expand that partition into an overlapping
cover.

\subsection{Basic theoretical properties}

We now record several basic properties of the proposed geometry. These results
are not meant to provide a complete statistical theory. Rather, they justify
the main geometric ingredients: finite-time Green profiles approximate exact
Green profiles under teleportation; normalization is stable away from zero;
the forward--backward cosine is a positive semidefinite kernel; and the
idealized disjoint and overlapping cosine assignment rules are exact under
deterministic margin conditions.

\begin{proposition}[Truncation error for the diffusive Green profile]
\label{prop:green-truncation-bound}
Let
\[
P=\alpha \widetilde P+(1-\alpha)\frac{1}{n}\mathbf 1\mathbf 1^\top,
\qquad 0<\alpha<1,
\]
where \(\widetilde P\) is row-stochastic. Let \(\pi\) be the stationary
distribution of \(P\), and let \(\Pi=\mathbf 1\pi^\top\). Define the infinite
diffusive Green profile by
\[
\mathcal G_{\geq 1}(P)
=
\sum_{t=1}^{\infty}(P^t-\Pi)
\]
and its truncation by
\[
\mathcal G_T^{\mathrm{diff}}(P)
=
\sum_{t=1}^{T}(P^t-\Pi).
\]
Then, for every vertex \(i\),
\[
\left\|
\mathcal G_{\geq 1}(P)_{i\bullet}
-
\mathcal G_T^{\mathrm{diff}}(P)_{i\bullet}
\right\|_1
\leq
\frac{2\alpha^{T+1}}{1-\alpha}.
\]
\end{proposition}

\begin{proof}
The truncation error is
\[
\mathcal G_{\geq 1}(P)_{i\bullet}
-
\mathcal G_T^{\mathrm{diff}}(P)_{i\bullet}
=
\sum_{t=T+1}^{\infty}
\left(P^t_{i\bullet}-\pi^\top\right).
\]
For a teleported chain, the Dobrushin contraction coefficient is at most
\(\alpha\), because every transition contains a common uniform component of
mass \(1-\alpha\). Hence
\[
\|P^t_{i\bullet}-\pi^\top\|_1\leq 2\alpha^t.
\]
Summing the geometric tail gives
\[
\left\|
\mathcal G_{\geq 1}(P)_{i\bullet}
-
\mathcal G_T^{\mathrm{diff}}(P)_{i\bullet}
\right\|_1
\leq
\sum_{t=T+1}^{\infty}2\alpha^t
=
\frac{2\alpha^{T+1}}{1-\alpha}.
\]
\end{proof}

\begin{remark}[Interpretation of the truncation bound]
The bound in Proposition~\ref{prop:green-truncation-bound} is a worst-case
total-variation bound. It is not intended to be a tight prediction of the
empirical error for the sparse directed benchmark graphs used in the
experiments. For the parameter values used in the experiments, this worst-case
bound should not be interpreted as a numerical guarantee that a small value of
\(T\) already gives a small approximation error. Its purpose is only to show
that the finite-time diffusive Green coordinate converges to the infinite
diffusive Green coordinate under teleportation.

In practice, the effective mixing inside and across planted communities may
lead to substantially smaller errors than this worst-case estimate. The
sensitivity of the method with respect to \(T\) and \(\alpha\) should therefore
be assessed empirically.
\end{remark}

\begin{lemma}[Stability of normalization]
\label{lem:normalization-stability}
Let \(x,y\in\mathbb R^d\) satisfy
\[
\|x-y\|_2\leq \varepsilon,
\qquad
\|x\|_2\geq r,
\qquad
\|y\|_2\geq r
\]
for some \(r>0\). Then
\[
\left\|
\frac{x}{\|x\|_2}
-
\frac{y}{\|y\|_2}
\right\|_2
\leq
\frac{2\varepsilon}{r}.
\]
\end{lemma}

\begin{proof}
We write
\[
\frac{x}{\|x\|_2}
-
\frac{y}{\|y\|_2}
=
\frac{x-y}{\|x\|_2}
+
y\left(
\frac{1}{\|x\|_2}
-
\frac{1}{\|y\|_2}
\right).
\]
Therefore
\[
\left\|
\frac{x}{\|x\|_2}
-
\frac{y}{\|y\|_2}
\right\|_2
\leq
\frac{\|x-y\|_2}{\|x\|_2}
+
\frac{|\|y\|_2-\|x\|_2|}{\|x\|_2}
\leq
\frac{2\varepsilon}{r}.
\]
\end{proof}

\begin{proposition}[Kernel property of forward--backward cosine]
\label{prop:fb-kernel}
Assume first that all one-sided Green profiles used in the normalization are
nonzero, so that
\[
\|\widehat x_u^+\|_2=\|\widehat x_u^-\|_2=1
\qquad
\text{for all }u\in V.
\]
Then, for every \(\lambda\in[0,1]\), the forward--backward coordinate
\[
z_u^{(\lambda)}
=
\left(
\sqrt{\lambda}\,\widehat x_u^{+},
\sqrt{1-\lambda}\,\widehat x_u^{-}
\right)
\]
has unit norm, and the matrix
\[
K_{\mathrm{FB}}
=
\left(\operatorname{cos}_{\mathrm{FB}}(u,v)\right)_{u,v\in V}
\]
is symmetric positive semidefinite. Moreover,
\[
\operatorname{cos}_{\mathrm{FB}}(u,v)
=
\lambda \cos^{+}(u,v)
+
(1-\lambda)\cos^{-}(u,v).
\]
Consequently,
\[
d_{\mathrm{FB}}(u,v)
=
\sqrt{2-2\operatorname{cos}_{\mathrm{FB}}(u,v)}
\]
is the Euclidean distance between the normalized forward--backward coordinates
\(z_u^{(\lambda)}\) and \(z_v^{(\lambda)}\).

In numerical implementations with the tolerance-based normalization in
\eqref{eq:fb-green-coordinate}, the positive semidefinite kernel property
still holds, because the resulting similarity matrix is a Gram matrix of the
computed coordinates. However, the exact convex-combination identity and the
Euclidean distance formula above are guaranteed only in the nondegenerate
unit-normalized case. This is the case covered by the theoretical statement;
the tolerance-based normalization is used only as a numerical safeguard for
nearly zero profiles.
\end{proposition}

\begin{proof}
Let \(Z_{\mathrm{FB}}\) be the matrix whose \(u\)-th row is
\(z_u^{(\lambda)}\). Then
\[
K_{\mathrm{FB}}=Z_{\mathrm{FB}}Z_{\mathrm{FB}}^\top.
\]
Thus \(K_{\mathrm{FB}}\) is symmetric positive semidefinite. The convex
combination formula follows from the definition of
\(z_u^{(\lambda)}\). Since \(\|z_u^{(\lambda)}\|_2=1\), we also have
\[
\|z_u^{(\lambda)}-z_v^{(\lambda)}\|_2^2
=
2-2\langle z_u^{(\lambda)},z_v^{(\lambda)}\rangle.
\]
\end{proof}

We next give a simple population-level separation statement. It says that in
an ideal directed block model, where vertices in the same planted community
have identical aggregate transition behavior, the proposed Green coordinates
are constant within planted communities.

Let
\[
\mathcal C^\star=\{C_1^\star,\ldots,C_K^\star\}
\]
be a planted disjoint partition of \(V\). A transition matrix \(P\) is called
block-constant with respect to \(\mathcal C^\star\) if, for every
\(u\in C_a^\star\) and \(v\in C_b^\star\),
\[
P_{uv}=\frac{B_{ab}}{|C_b^\star|}
\]
for some row-stochastic \(K\times K\) matrix \(B=(B_{ab})\).

\begin{theorem}[Population block separation for diffusive Green coordinates]
\label{thm:population-block-separation}
Assume that \(P_\alpha^+\) and \(P_\alpha^-\) are block-constant with respect
to the same planted partition
\[
\mathcal C^\star=\{C_1^\star,\ldots,C_K^\star\}.
\]
Assume also that the one-sided diffusive Green profiles
\[
\mathcal G_T^{\mathrm{diff}}(P_\alpha^+)_{u\bullet},
\qquad
\mathcal G_T^{\mathrm{diff}}(P_\alpha^-)_{u\bullet}
\]
used in the normalization are nonzero. Then the forward--backward diffusive
Green coordinate \(z_u^{(\lambda)}\) is constant inside each planted
community. Hence, if \(u,v\in C_a^\star\), then
\[
z_u^{(\lambda)}=z_v^{(\lambda)}
\qquad\text{and}\qquad
\operatorname{cos}_{\mathrm{FB}}(u,v)=1.
\]
If the \(K\) community-level forward--backward diffusive Green coordinates are
pairwise distinct, then
\[
\operatorname{cos}_{\mathrm{FB}}(u,v)<1
\]
whenever \(u\in C_a^\star\), \(v\in C_b^\star\), and \(a\neq b\).
\end{theorem}

\begin{proof}
We prove the statement for a generic transition matrix \(P\) that is
block-constant with respect to \(\mathcal C^\star\). By definition, for every
\(u\in C_a^\star\) and \(v\in C_b^\star\),
\[
P_{uv}=\frac{B_{ab}}{|C_b^\star|}
\]
for some row-stochastic block transition matrix \(B=(B_{ab})\). Therefore all
vertices in the same planted community have identical transition
probabilities to every target block, with the probability mass distributed
uniformly inside each target block.

It follows by induction that, for every \(t\geq 1\), the matrix \(P^t\) is
also block-constant with respect to \(\mathcal C^\star\). Moreover, the
stationary distribution is uniform inside each block after conditioning on
that block: if
\[
\bar\pi_b=\sum_{v\in C_b^\star}\pi_v,
\]
then
\[
\pi_v=\frac{\bar\pi_b}{|C_b^\star|}
\qquad
\text{for every }v\in C_b^\star.
\]
Hence the stationary matrix \(\Pi=\mathbf 1\pi^\top\) is also block-constant.

Consequently, for every \(t\geq 1\), the rows of \(P^t-\Pi\) are identical
for all vertices belonging to the same planted community. Summing from
\(t=1\) to \(T\), we obtain that
\[
\mathcal G_T^{\mathrm{diff}}(P)_{u\bullet}
=
\mathcal G_T^{\mathrm{diff}}(P)_{v\bullet}
\]
whenever \(u,v\in C_a^\star\).

Applying this argument to both \(P_\alpha^+\) and \(P_\alpha^-\), we get
\[
x_u^+=x_v^+,
\qquad
x_u^-=x_v^-,
\]
for all \(u,v\in C_a^\star\). Since the profiles are assumed to be nonzero,
normalization preserves equality, and hence
\[
z_u^{(\lambda)}=z_v^{(\lambda)}.
\]
Therefore the within-community forward--backward cosine equals \(1\).

Finally, if two community-level normalized coordinates are pairwise distinct,
then their Euclidean distance is positive. Since all coordinates have unit
norm, their inner product must be strictly smaller than \(1\). This proves the
between-community statement.
\end{proof}

\begin{remark}
The block-constant assumption is an ideal population condition. It implies
that vertices in the same planted community have identical aggregate
transition behavior and that transition mass is distributed uniformly inside
each target block. This assumption is not intended to cover degree-corrected
block models directly, where vertices in the same community may have different
in-degree and out-degree propensities. The theorem should therefore be read as
a population separation sanity check for the Green coordinate, rather than as
a full statistical recovery theorem for all benchmark models considered in
the experiments.
\end{remark}

The preceding theorem is a population statement. The next two results give
deterministic margin conditions for the idealized assignment rules underlying
the two algorithmic stages. They should not be read as global convergence
guarantees for the nonconvex spherical \(K\)-means optimization. Rather, they
state that, once the appropriate community representatives and thresholds are
fixed, the corresponding cosine assignment rules are exact under a positive
margin.

Let \(z_u=z_u^{(\lambda)}\). For a disjoint partition
\[
\mathcal C^\star=\{C_1^\star,\ldots,C_K^\star\},
\]
assume that
\[
\sum_{u\in C_a^\star}z_u\neq 0
\qquad
\text{for every }a=1,\ldots,K.
\]
Define the normalized community centroid
\[
m_a
=
\frac{\sum_{u\in C_a^\star}z_u}
{\left\|\sum_{u\in C_a^\star}z_u\right\|_2}.
\]

\begin{proposition}[Nearest-centroid recovery under a cosine margin]
\label{prop:cosine-margin-disjoint}
Assume that there exists \(\gamma>0\) such that, for every
\(u\in C_a^\star\) and every \(b\neq a\),
\[
\langle z_u,m_a\rangle
\geq
\langle z_u,m_b\rangle+\gamma.
\]
Then the ideal nearest-centroid cosine rule
\[
\widehat c(u)
=
\argmax_{1\leq b\leq K}\langle z_u,m_b\rangle
\]
recovers the planted label of every vertex.
\end{proposition}

\begin{proof}
For \(u\in C_a^\star\), the assumed margin implies
\[
\langle z_u,m_a\rangle>
\langle z_u,m_b\rangle
\]
for every \(b\neq a\). Hence the maximizer is uniquely \(a\).
\end{proof}

\begin{remark}
This proposition is an assignment-level guarantee. It assumes the
community representatives \(m_a\) associated with the planted partition are
already fixed. It does not claim that a particular run of the nonconvex
spherical \(K\)-means algorithm necessarily finds these representatives from
arbitrary initialization.
\end{remark}

The following result is stated for the overlap rule that will be introduced in
\cref{sec:green-fb-overlap}. Here \(s_j(u)\) denotes the vertex--community
cosine score between vertex \(u\) and the candidate community \(C_j\), and
\(\theta_j\) denotes the community-adaptive threshold used by that rule.

For this idealized statement, assume that the candidate communities
\(C_1,\ldots,C_K\) have been indexed consistently with the \(K\) true
community labels. Let
\[
    M^\star(u)\subseteq \{1,\ldots,K\}
\]
denote the true membership-label set of vertex \(u\). Assume that the initial
disjoint label \(b(u)\) satisfies
\[
    b(u)\in M^\star(u),
\]
that is, the initial partition assigns each vertex to one of its true
communities.

\begin{theorem}[Exact overlap recovery under a margin condition]
\label{thm:overlap-margin-recovery}
Assume that there exists \(\gamma>0\) such that, for every vertex \(u\) and
every community index \(j\in\{1,\ldots,K\}\),
\[
s_j(u)\geq \theta_j+\gamma
\qquad
\text{if } j\in M^\star(u),
\]
and
\[
s_j(u)\leq \theta_j-\gamma
\qquad
\text{if } j\notin M^\star(u).
\]
Assume also that \(b(u)\in M^\star(u)\) for every vertex \(u\). Then the
Di-Green-FB-Cosine Overlap algorithm exactly recovers the true membership
labels:
\[
\widehat M(u)=M^\star(u)
\qquad
\text{for all }u\in V,
\]
where \(\widehat M(u)\) denotes the set of community indices returned by the
algorithm.
\end{theorem}

\begin{proof}
The algorithm always includes the initial community index \(b(u)\), which is a
true membership by assumption. If \(j\in M^\star(u)\), then
\[
s_j(u)\geq \theta_j+\gamma>\theta_j,
\]
so the algorithm adds community \(j\). If \(j\notin M^\star(u)\), then
\[
s_j(u)\leq \theta_j-\gamma<\theta_j,
\]
so the algorithm does not add community \(j\). Therefore every true
membership is added and no false membership is added.
\end{proof}

\begin{corollary}[Stability of overlap recovery under score and threshold perturbations]
\label{cor:overlap-stability}
Suppose the margin condition in
\cref{thm:overlap-margin-recovery} holds with margin \(\gamma>0\) for the
exact Green cosine scores \(s_j(u)\) and thresholds \(\theta_j\). Let
\(\widetilde s_j(u)\) and \(\widetilde\theta_j\) be the scores and thresholds
computed from approximate or truncated Green coordinates. If
\[
    |\widetilde s_j(u)-s_j(u)|
    +
    |\widetilde\theta_j-\theta_j|
    <\gamma
\]
for all vertices \(u\) and all community indices \(j\), then the overlap
assignment obtained from the approximate scores and thresholds is identical
to the assignment obtained from the exact scores and thresholds.
\end{corollary}

\begin{proof}
If \(j\in M^\star(u)\), then the margin condition gives
\[
s_j(u)-\theta_j\geq \gamma.
\]
Therefore
\[
\widetilde s_j(u)-\widetilde\theta_j
\geq
s_j(u)-\theta_j
-
|\widetilde s_j(u)-s_j(u)|
-
|\widetilde\theta_j-\theta_j|
>0.
\]
Thus the approximate rule accepts community \(j\).

Similarly, if \(j\notin M^\star(u)\), then
\[
s_j(u)-\theta_j\leq -\gamma.
\]
Hence
\[
\widetilde s_j(u)-\widetilde\theta_j
\leq
s_j(u)-\theta_j
+
|\widetilde s_j(u)-s_j(u)|
+
|\widetilde\theta_j-\theta_j|
<0.
\]
Thus the approximate rule rejects community \(j\). Consequently every
accept/reject decision is unchanged.
\end{proof}

\section{Disjoint Community Detection via Forward--Backward Green Coordinates}
\label{sec:green-fb-disjoint}

The first stage of the proposed framework constructs a non-overlapping
partition. Although any disjoint community detection method could be used as
the initializer, the main pipeline in this paper uses
Di-Green-FB-cosine-KMeans, which clusters the forward--backward Green
coordinates directly. This keeps the disjoint and overlap stages within the
same geometric representation.

Given the forward and backward Green coordinates \(x_u^{+}\) and \(x_u^{-}\),
we first normalize them as in \eqref{eq:green-normalization}, and then form
the preliminary combined coordinate
\[
\widetilde z_u^{(\lambda)}
=
\left(
\sqrt{\lambda}\,\widehat x_u^{+},
\sqrt{1-\lambda}\,\widehat x_u^{-}
\right),
\qquad 0\leq \lambda\leq 1.
\]
In the nondegenerate case, \(\|\widetilde z_u^{(\lambda)}\|_2=1\), and we set
\(z_u^{(\lambda)}=\widetilde z_u^{(\lambda)}\). In numerical implementation,
we use the tolerance-based final normalization described in
\eqref{eq:fb-green-coordinate}. Spherical \(K\)-means is then applied to the
resulting coordinates \(z_u^{(\lambda)}\). The resulting clusters form the
initial disjoint partition.

The complete procedure is summarized in
Algorithm~\ref{alg:green-fb-kmeans}.

\begin{algorithm}[H]
\small
\DontPrintSemicolon
\SetAlgoLined
\SetKwInOut{KwInput}{Input}
\SetKwInOut{KwOutput}{Output}

\caption{Di-Green-FB-cosine-KMeans}
\label{alg:green-fb-kmeans}

\KwInput{
\begin{minipage}[t]{0.82\linewidth}
Directed weighted graph \(G=(V,E,A)\); number of communities \(K\);
teleportation parameter \(\alpha\); truncation length \(T\);
forward--backward weight \(\lambda\in[0,1]\); normalization tolerance
\(\tau>0\).
\end{minipage}
}

\KwOutput{
\begin{minipage}[t]{0.82\linewidth}
A disjoint partition \(\mathcal P=\{C_1,\ldots,C_K\}\).
\end{minipage}
}

\BlankLine

Construct the forward teleported transition matrix \(P_\alpha^{+}\) from
\(A\)\;

Construct the backward teleported transition matrix \(P_\alpha^{-}\) from
\(A^\top\)\;

Compute stationary distributions \(\pi^{+}\) and \(\pi^{-}\)\;

Compute diffusive truncated Green matrices
\[
\mathcal G_T^{\mathrm{diff}}(P_\alpha^{+})
=
\sum_{t=1}^{T}
\left((P_\alpha^{+})^t-\mathbf 1(\pi^{+})^\top\right)
\]
and
\[
\mathcal G_T^{\mathrm{diff}}(P_\alpha^{-})
=
\sum_{t=1}^{T}
\left((P_\alpha^{-})^t-\mathbf 1(\pi^{-})^\top\right).
\]

\ForEach{\(u\in V\)}{
Set
\(x_u^{+}\leftarrow \mathcal G_T^{\mathrm{diff}}(P_\alpha^{+})_{u\bullet}\)\;

Set
\(x_u^{-}\leftarrow \mathcal G_T^{\mathrm{diff}}(P_\alpha^{-})_{u\bullet}\)\;

    Normalize, with a small tolerance \(\tau>0\),
    \[
    \widehat x_u^{+}
    \leftarrow
    \frac{x_u^{+}}{\max\{\|x_u^{+}\|_2,\tau\}},
    \qquad
    \widehat x_u^{-}
    \leftarrow
    \frac{x_u^{-}}{\max\{\|x_u^{-}\|_2,\tau\}}.
    \]

    Form the preliminary coordinate
    \[
    \widetilde z_u^{(\lambda)}
    \leftarrow
    \left(
    \sqrt{\lambda}\,\widehat x_u^{+},
    \sqrt{1-\lambda}\,\widehat x_u^{-}
    \right).
    \]

    Apply a final normalization
    \[
    z_u^{(\lambda)}
    \leftarrow
    \frac{\widetilde z_u^{(\lambda)}}
    {\max\{\|\widetilde z_u^{(\lambda)}\|_2,\tau\}}.
    \]
}

Apply spherical \(K\)-means to the normalized coordinates
\(\{z_u^{(\lambda)}:u\in V\}\)\;

\Return{the resulting partition \(\mathcal P=\{C_1,\ldots,C_K\}\)}\;

\end{algorithm}

\section{Overlap Detection via Green Forward--Backward Cosine Similarity}
\label{sec:green-fb-overlap}

This section describes the second stage of the framework: expanding an initial
disjoint partition
\[
\mathcal P=\{C_1,C_2,\ldots,C_K\}
\]
into an overlapping cover. The partition may be obtained by any disjoint
community detection algorithm; in the main pipeline, it is produced by
Di-Green-FB-cosine-KMeans. Given the vertex coordinates
\(z_u^{(\lambda)}\) and the forward--backward cosine
\(\operatorname{cos}_{\mathrm{FB}}(u,v)\) defined in
Section~\ref{sec:green-geometry}, the goal is to identify vertices whose
diffusion profiles are sufficiently aligned with more than one community.

A second important feature of our method is that the overlap criterion is
community-adaptive. Rather than using a single global cosine threshold for all
communities, we estimate a separate threshold for each community. This makes
the overlap rule more flexible for communities with different sizes, densities,
and internal coherence.

\subsection{Community-wise cosine thresholds}

For each community \(C_j\), we compute its internal forward--backward cosine
distribution:
\begin{equation}
\mathcal S_j
=
\left\{
\operatorname{cos}_{\mathrm{FB}}(u,v)
:
u,v\in C_j,\ u<v
\right\}.
\label{eq:internal-cosine-distribution}
\end{equation}
If \(|C_j|<2\), then \(\mathcal S_j\) is empty. In that case, we set
\[
\theta_j=\theta_{\min}.
\]
Thus the quantile-based threshold is used only for communities containing at
least two vertices.

For each community, the adaptive threshold is defined by
\begin{equation}
\theta_j
=
\max
\left\{
\theta_{\min},
Q_q(\mathcal S_j)-\delta+\varepsilon
\right\}.
\label{eq:community-threshold}
\end{equation}
Here \(Q_q(\mathcal S_j)\) denotes the \(q\)-quantile of the internal cosine
values in \(C_j\), \(\delta\ge 0\) is a relaxation margin, \(\varepsilon\) is
a small offset, and \(\theta_{\min}\) is a lower bound preventing the
threshold from becoming too permissive.

The use of \(\theta_j\) has an important advantage. Dense and coherent
communities typically have high internal cosine values and therefore receive
higher thresholds. Sparse or heterogeneous communities receive lower
thresholds. Hence the overlap rule adapts automatically to the local geometry
of each detected community.

\subsection{Cosine-based overlap assignment rule}

For a vertex \(u\in V\) and a community \(C_j\), we define the similarity
between \(u\) and \(C_j\) by averaging the largest cosine similarities between
\(u\) and vertices of \(C_j\). More precisely, let
\[
\ell_j
=
\max\{1,\lceil \eta |C_j|\rceil\},
\qquad 0<\eta\le 1,
\]
and let \(\operatorname{Top}_{\ell_j}(u,C_j)\) be the set of the \(\ell_j\)
vertices in \(C_j\) with the largest values of
\(\operatorname{cos}_{\mathrm{FB}}(u,v)\). We define
\begin{equation}
s_j(u)
=
\frac{1}{\ell_j}
\sum_{v\in \operatorname{Top}_{\ell_j}(u,C_j)}
\operatorname{cos}_{\mathrm{FB}}(u,v).
\label{eq:vertex-community-cosine}
\end{equation}

The parameter \(\eta\) controls the robustness of the vertex--community
similarity. When \(\eta\) is very small, the rule approaches a maximum-cosine
criterion. Larger values of \(\eta\) require the vertex to be similar to a
larger portion of the community and thus yield a more stable assignment.

Let \(b(u)\) denote the initial community label of \(u\), that is,
\(u\in C_{b(u)}\). The final overlapping membership set of \(u\) is defined by
\begin{equation}
\widehat{\mathcal C}(u)
=
\{C_{b(u)}\}
\cup
\left\{
C_j:
j\ne b(u),\
s_j(u)\ge \theta_j
\right\}.
\label{eq:overlap-assignment}
\end{equation}
Thus, the initial membership of each vertex is always preserved, and
additional memberships are added only when the similarity between the vertex
and another community exceeds that community's adaptive threshold.

Importantly, the acceptance condition in \eqref{eq:overlap-assignment} is
purely geometric: it depends only on the forward--backward Green cosine
between the vertex and the candidate community. No modularity gain, density
gain, or degree-based acceptance condition is imposed in the final overlap
assignment rule. This makes the refinement step directly interpretable as an
assignment in the proposed Green cosine geometry.

\subsection{The Di-Green-FB-Cosine Overlap Algorithm}

The complete procedure is summarized in
Algorithm~\ref{alg:green-fb-overlap-compact}.

\begin{algorithm}[H]
\small
\DontPrintSemicolon
\SetAlgoLined
\SetKwInOut{KwInput}{Input}
\SetKwInOut{KwOutput}{Output}

\caption{Di-Green-FB-Cosine Overlap Detection}
\label{alg:green-fb-overlap-compact}

\KwInput{
\begin{minipage}[t]{0.82\linewidth}
Initial disjoint partition \(\mathcal P=\{C_1,\ldots,C_K\}\), obtained by any
base disjoint community detection algorithm; vertex coordinates
\(z_u^{(\lambda)}\) for all \(u\in V\); parameters
\(q,\delta,\theta_{\min},\varepsilon,\eta\).
\end{minipage}
}

\KwOutput{
\begin{minipage}[t]{0.82\linewidth}
Overlapping assignment \(\widehat{\mathcal C}(u)\), \(u\in V\).
\end{minipage}
}

\BlankLine

Compute
\(\operatorname{cos}_{\mathrm{FB}}(u,v)
=
\langle z_u^{(\lambda)},z_v^{(\lambda)}\rangle\)
for all pairs \(u,v\in V\)\;

\ForEach{\(C_j\in\mathcal P\)}{
    \eIf{\(|C_j|<2\)}{
        Set \(\theta_j\leftarrow \theta_{\min}\)\;
    }{
        Compute the internal cosine set \(\mathcal S_j\)\;
        
        Set
        \[
        \theta_j
        \leftarrow
        \max\{\theta_{\min},Q_q(\mathcal S_j)-\delta+\varepsilon\}.
        \]
    }
}

\ForEach{\(u\in V\)}{
    Initialize \(\widehat{\mathcal C}(u)\leftarrow\{C_{b(u)}\}\)\;

    \ForEach{\(C_j\in\mathcal P\), \(j\neq b(u)\)}{
        Compute the top-\(\ell_j\) average cosine \(s_j(u)\)\;

        \If{\(s_j(u)\ge \theta_j\)}{
            Add \(C_j\) to \(\widehat{\mathcal C}(u)\)\;
        }
    }
}

\Return{\(\widehat{\mathcal C}(u)\) for all \(u\in V\)}\;

\end{algorithm}

\subsection{Computational complexity}
\label{subsec:complexity}

Let \(n=|V|\), let \(m_{\mathrm{nz}}\) be the number of nonzero directed edges,
let \(K\) be the number of initial communities, \(T\) be the Green truncation
length, and \(d\) be the dimension of the final forward--backward coordinate. In the full-coordinate implementation used in
this paper, each one-sided diffusive Green profile has dimension \(n\), and
therefore \(d=2n\). Storing the forward and backward Green coordinates
requires \(O(n^2)\) memory.

The cost of constructing full diffusive Green coordinates depends on how the
truncated sums
\[
\mathcal G_T^{\mathrm{diff}}(P)
=
\sum_{t=1}^{T}(P^t-\mathbf 1\pi^\top)
\]
are computed. Although the teleported matrix
\[
P_\alpha=\alpha\widetilde P+(1-\alpha)\frac{1}{n}\mathbf 1\mathbf 1^\top
\]
is dense because of the rank-one teleportation term, this term can be applied
separately. Hence multiplying a dense \(n\times n\) matrix by \(P_\alpha\) can
be implemented using one sparse multiplication by \(\widetilde P\) plus one
rank-one correction. This costs \(O(m_{\mathrm{nz}}n+n^2)\), which is \(O(m_{\mathrm{nz}}n)\) when
\(m_{\mathrm{nz}}\geq n\). Thus, computing the full forward and backward
coordinates costs
\[
O\bigl(T(m_{\mathrm{nz}}n+n^2)\bigr),
\]
or \(O(Tm_{\mathrm{nz}}n)\) in the usual sparse-graph regime with
\(m_{\mathrm{nz}}\geq n\).

For the disjoint clustering stage, an all-pairs cosine matrix is not required.
Spherical \(K\)-means can be applied directly to the \(n\) normalized
coordinates. If \(I\) denotes the number of \(K\)-means iterations, this step
costs
\[
O(I n K d).
\]
In the full-coordinate case \(d=2n\), this becomes \(O(IKn^2)\), which is
substantially smaller than explicitly forming all pairwise similarities when
\(K\ll n\).

The overlap expansion stage is more expensive in a direct implementation,
because it compares vertices with candidate communities using cosine
similarities. If all pairwise cosine similarities are precomputed, this step
costs
\[
O(n^2d)
\]
time and \(O(n^2)\) memory for the similarity matrix. With full
forward--backward Green coordinates, \(d=2n\), giving \(O(n^3)\) time for this
direct all-pairs implementation. Once the similarity matrix is available,
community-wise thresholds and top-\(\ell_j\) averages can be computed in at
most \(O(n^2\log n)\) time, or \(O(n^2)\) time using partial selection.

Therefore, the full implementation is best viewed as a proof-of-concept and
benchmark implementation for small and medium-size directed networks. Larger
graphs require reduced or approximate representations, such as truncated
singular-vector Green coordinates, landmark Green coordinates, random
projections, sketching, or approximate nearest-neighbor search in the Green
feature space. These approximations reduce the effective dimension from
\(d=2n\) to \(d\ll n\), and can also avoid storing the full \(n\times n\)
cosine matrix.

\section{Experiments}
\label{sec:experiments}

This section evaluates the proposed forward--backward Green cosine framework
in two tasks: disjoint directed community detection and overlapping directed
community detection. To avoid repetition, we first describe the common
experimental protocol, including benchmark models, evaluation metrics,
compared methods, and reporting conventions. Each subsequent experiment then
only specifies the particular graph regime being tested and reports the
corresponding numerical results.

Most experiments are conducted on synthetic directed benchmark networks with
known ground-truth communities. This controlled setting allows us to vary the
mixing parameter, graph size, number of communities, community-size
heterogeneity, degree heterogeneity, and overlap rate. We also include a
modularity-only experiment on real directed networks. Since reliable disjoint
or overlapping ground-truth communities are not available for all real
networks considered, this real-data experiment is used only as an internal
structural-quality evaluation, not as a ground-truth recovery test.

The synthetic experiments use three benchmark families. The first family is a
directed heterogeneous Gaussian partition model, which tests the method under
heterogeneous community sizes and heterogeneous block densities. The second is
a directed degree-corrected block model, which tests robustness to
heterogeneous in-degree and out-degree propensities. This model is a directed
version of the degree-corrected stochastic block model, introduced to separate
community structure from degree heterogeneity
\cite{KarrerNewman2011,QinRohe2013,WangLiangJi2020}. The third family is a
directed overlapping planted-partition model, which is used to evaluate the
overlap expansion stage.

The main tables focus on accuracy and structural quality. Running times are
not included in the main experimental tables, because they depend strongly on
implementation details and hardware. The computational cost of the full
Green-coordinate implementation is discussed separately in
Section~\ref{subsec:complexity}.

\subsection{Evaluation metrics}
\label{subsec:evaluation-metrics}

For disjoint community detection, we report normalized mutual information
(NMI) \cite{StrehlGhosh2002}, adjusted Rand index (ARI)
\cite{HubertArabie1985}, pairwise \(F_1\), and directed modularity
\(Q_{\mathrm{dir}}\) \cite{LeichtNewman2008}. NMI, ARI, and PairF1 measure
agreement with the planted partition, while \(Q_{\mathrm{dir}}\) evaluates the
internal structural quality of the detected directed communities.

For a directed graph with adjacency matrix \(A\), let
\[
m=\sum_{i,j}A_{ij},
\qquad
k_i^{\mathrm{out}}=\sum_j A_{ij},
\qquad
k_j^{\mathrm{in}}=\sum_i A_{ij}.
\]
For a disjoint partition with labels \(c_i\), we use the Leicht--Newman
directed modularity
\begin{equation}
Q_{\mathrm{dir}}
=
\frac{1}{m}
\sum_{i,j}
\left(
A_{ij}
-
\frac{k_i^{\mathrm{out}}k_j^{\mathrm{in}}}{m}
\right)
\mathbf 1\{c_i=c_j\}.
\label{eq:directed-modularity}
\end{equation}
Thus, unlike undirected modularity, the normalization factor is \(1/m\), not
\(1/(2m)\).

For overlapping community detection, we report overlapping normalized mutual
information (ONMI) \cite{McDaidGreeneHurley2011}, PairF1, and OverlapF1. ONMI compares two overlapping
covers, PairF1 measures whether pairs of vertices co-occur in at least one
community, and OverlapF1 evaluates the ability to identify vertices with more
than one membership. We also use the aggregate score
\[
    \mathrm{Score}
    =
    \frac{\mathrm{ONMI}+\mathrm{PairF1}+\mathrm{OverlapF1}}{3}
\]
as a compact summary of overlap performance.

\subsection{Synthetic directed benchmark models}
\label{subsec:synthetic-models}

We use three families of synthetic directed graphs.

\paragraph{Directed heterogeneous Gaussian partition graphs.}
For disjoint community detection, we use directed heterogeneous Gaussian
partition graphs. The parameters are the number of vertices \(N\), the number
of planted communities \(K\), the target average out-degree \(\bar d\), the
mixing parameter \(\mu\), and a community-size heterogeneity parameter. 

The community sizes are sampled from a Gaussian-like distribution and then
normalized so that their total is \(N\). More precisely, we first generate
positive community-size weights with the prescribed heterogeneity level,
rescale them so that their sum is \(N\), and then round them to integer sizes
while preserving the total number of vertices. Very small communities, if any,
are adjusted to satisfy the required minimum size.

Edges are directed and generated independently. Self-loops are not allowed,
and the two directed edges \(i\to j\) and \(j\to i\) are sampled
independently. For a vertex in community \(C_c\), the within-community edge
probability is
\[
    p_{\mathrm{in}}^{(c)}
    =
    \rho_c
    \frac{(1-\mu)\bar d}{|C_c|-1},
\]
where \(\rho_c\) is a community-specific internal density multiplier. The
inter-community edge probability is
\[
    p_{\mathrm{out}}^{(c)}
    =
    \frac{\mu\bar d}{N-|C_c|}.
\]
All probabilities are clipped to the interval \([0,1]\) if necessary. In the
parameter regimes used in our experiments, clipping is rare and the realized
average degree is reported by the graph generator.

Because of the multiplier \(\rho_c\), the parameter \(\mu\) should be
interpreted as a nominal mixing parameter rather than an exact realized
external-edge fraction. Indeed, before clipping, the expected internal
out-degree scale for a vertex in \(C_c\) is
\[
    \rho_c(1-\mu)\bar d,
\]
whereas the expected external out-degree scale is
\[
    \mu\bar d.
\]
Thus the corresponding expected external fraction for community \(C_c\) is
approximately
\[
    \frac{\mu}{\rho_c(1-\mu)+\mu},
\]
not exactly \(\mu\) unless \(\rho_c=1\). Nevertheless, increasing \(\mu\)
monotonically weakens the planted community separation and therefore remains
a meaningful way to control the difficulty of the benchmark.

Unless stated otherwise, the internal density multipliers are sampled from
\[
    \rho_c\in[0.65,1.55].
\]
Consequently, both the target average degree \(\bar d\) and the mixing
parameter \(\mu\) should be interpreted as nominal graph-generation
parameters rather than exact deterministic graph statistics.

\paragraph{Directed degree-corrected block graphs.}
For disjoint community detection, we also use directed degree-corrected block
graphs. This benchmark is designed to test robustness to heterogeneous
outgoing and incoming degrees. Degree heterogeneity is a major difficulty in
community detection, because vertices with large degrees may dominate spectral
or local-density-based representations even when the underlying community
structure is unchanged. The degree-corrected stochastic block model was
introduced to address this issue in undirected networks
\cite{KarrerNewman2011,QinRohe2013}, and directed degree-corrected block
models have been used as standard benchmarks for directed spectral methods
\cite{WangLiangJi2020}.

Each vertex \(i\) has a planted community label \(c_i\), an outgoing
degree-correction parameter \(\theta_i^{\mathrm{out}}\), and an incoming
degree-correction parameter \(\theta_i^{\mathrm{in}}\). Conditional on these
parameters, directed edges are generated independently according to
\[
    \mathbb P(A_{ij}=1)
    =
    \theta_i^{\mathrm{out}}
    \theta_j^{\mathrm{in}}
    B_{c_i c_j},
\]
where \(B=(B_{ab})\) is a block connectivity matrix. The parameters
\(\theta_i^{\mathrm{out}}\) and \(\theta_i^{\mathrm{in}}\) control the
propensity of vertex \(i\) to send and receive edges, respectively. In our implementation, the outgoing and incoming degree-correction parameters
are sampled independently and then normalized within each planted community:
\[
\sum_{i\in C_a}\theta_i^{\mathrm{out}}=|C_a|,
\qquad
\sum_{i\in C_a}\theta_i^{\mathrm{in}}=|C_a|.
\]
This normalization keeps the block matrix \(B\) responsible for the expected
mixing pattern between communities, while allowing substantial heterogeneity
among vertices inside the same community. Edge probabilities are clipped to
\([0,1]\) if necessary.

In the implementation used in the experiments, the block matrix is chosen as
\[
    B_{aa}
    =
    \frac{(1-\mu)\bar d}{|C_a|-1},
    \qquad
    B_{ab}
    =
    \frac{\mu\bar d}{N-|C_a|}
    \quad (a\neq b),
\]
before clipping the resulting edge probabilities to \([0,1]\). Thus, after
the within-community normalization of
\(\theta^{\mathrm{out}}\) and \(\theta^{\mathrm{in}}\), the expected
within-community and between-community out-degree scales are approximately
\((1-\mu)\bar d\) and \(\mu\bar d\), respectively. Larger \(\mu\) therefore
corresponds to weaker planted community separation. 
Compared with the Gaussian partition benchmark, this model introduces an
additional source of difficulty: vertices in the same community may have very
different outgoing and incoming degrees. This makes it a useful test of
whether the proposed forward--backward Green geometry is robust to directed
degree heterogeneity.

\paragraph{Directed overlapping planted-partition graphs.}
For overlapping community detection, each vertex has one primary community,
and a prescribed fraction of vertices receive one additional membership. The
parameters are the number of vertices \(N\), number of base communities \(K\),
target average out-degree \(\bar d\), mixing parameter \(\mu\), number of
overlapping vertices \(o_n\), and number of memberships per overlapping vertex
\(o_m\).

Let \(M(u)\) be the planted membership set of vertex \(u\). For each source
vertex \(u\), define
\[
    N_{\mathrm{in}}(u)
    =
    \left|\{v\neq u:\ M(u)\cap M(v)\neq\emptyset\}\right|,
\]
and
\[
    N_{\mathrm{out}}(u)
    =
    \left|\{v:\ M(u)\cap M(v)=\emptyset\}\right|.
\]
A directed edge \(u\to v\) is generated independently with probability
\[
p(u,v)
=
\begin{cases}
\dfrac{(1-\mu)\bar d}{N_{\mathrm{in}}(u)},
& M(u)\cap M(v)\neq\emptyset,\\[10pt]
\dfrac{\mu\bar d}{N_{\mathrm{out}}(u)},
& M(u)\cap M(v)=\emptyset.
\end{cases}
\]
Self-loops are not allowed, and probabilities are clipped to \([0,1]\) if
necessary. With this construction, the expected out-degree of each vertex is
approximately \(\bar d\), with an approximate fraction \(\mu\) of its outgoing
edges going to vertices with no shared planted membership. Unless stated
otherwise, we use
\[
    o_n=0.15N,
    \qquad
    o_m=2.
\]

\subsection{Compared Algorithms}
\label{subsec:compared-algorithms}

\paragraph{Methods for disjoint community detection.}
We compare the proposed method with directed spectral embedding methods and a
flow-based reference method. In the diagnostic experiment D0, we additionally
compare raw hitting-time cosine with centered Green cosine variants. 
The main comparison focuses on \(K\)-based disjoint community detection
methods, where the planted number of communities \(K\) is given to the
algorithm. This is a standard controlled setting for synthetic benchmark
experiments, since it isolates the quality of the vertex representation and
clustering geometry from the separate problem of model selection.

We note, however, that the current Di-Green-FB-cosine-KMeans implementation
does not automatically estimate \(K\). In real applications, \(K\) must either
be chosen by a model-selection criterion, selected by sweeping a candidate
grid, or provided by prior knowledge. Automatic estimation of the number of
communities is left for future work.

Directed Louvain and directed Leiden variants are natural
modularity-optimization references. They are not included in the main
\(K\)-controlled embedding comparison because they optimize a different
objective and determine the partition structure through modularity dynamics,
whereas the purpose of the main synthetic comparison is to isolate the quality
of vertex representations under a prescribed \(K\). We therefore do not use
them as direct representation-learning baselines. A broader comparison with
directed modularity-optimization methods is an important direction for future
work.

For the proposed method and the spectral baselines, the final partition is
obtained by applying \(K\)-means or spherical \(K\)-means to a vertex
representation. We use \(K\)-means++ initialization \cite{ArthurVassilvitskii2007}.
For cosine-normalized coordinates, we use spherical \(K\)-means, following the
standard spherical clustering framework \cite{DhillonModha2001}.

\begin{itemize}
    \item \textbf{Di-Green-FB-cosine-KMeans}. This is the proposed method for
    disjoint community detection. It constructs the forward--backward
    truncated Green coordinates introduced in
    Section~\ref{sec:green-fb-coordinates}, normalizes the forward and
    backward profiles, concatenates them with weight \(\lambda\), and applies
    spherical \(K\)-means as described in
    Algorithm~\ref{alg:green-fb-kmeans}.

    \item \textbf{oPCA}. This is a directed spectral baseline based on the
    leading left and right singular vectors of the adjacency matrix. Given the
    directed adjacency matrix \(A\), we compute its leading singular vectors,
    concatenate the left and right embeddings, and apply \(K\)-means. This
    baseline follows the PCA-type directed spectral comparison used in
    \cite{WangLiangJi2020}.

    \item \textbf{rPCA}. This is the regularized version of oPCA. Instead of
    applying singular value decomposition directly to \(A\), it first
    constructs a regularized directed Laplacian or regularized normalized
    adjacency matrix, and then applies the same singular-vector embedding and
    \(K\)-means procedure. This regularization is designed to reduce the
    effect of degree heterogeneity in sparse directed networks
    \cite{KimShi2012,WangLiangJi2020}.

    \item \textbf{D-SCORE}. This is a directed spectral ratio method. It uses
    the leading left and right singular vectors of the adjacency matrix and
    forms ratio-type features before clustering. The ratio transformation is
    intended to reduce the influence of degree heterogeneity in directed
    degree-corrected block models \cite{WangLiangJi2020}.

    \item \textbf{D-SCOREq}. This is a row-normalized variant of D-SCORE. It
    replaces the entrywise ratio step by row normalization of the singular
    vector matrices, using an \(\ell_q\)-normalization before applying
    \(K\)-means. This method is also taken from the directed spectral
    framework of \cite{WangLiangJi2020}.



    \item \textbf{Directed Infomap}. We also report Directed Infomap as a
    flow-based reference method \cite{RosvallBergstrom2008}. Infomap is not a
    \(K\)-based embedding-clustering method: it optimizes a coding objective
    for random-walk flow and automatically determines the number of
    communities. Therefore, we report the number of communities returned by Infomap together
with the quality metrics and treat it as a supplementary flow-based reference.
\end{itemize}

\paragraph{Methods for overlapping community detection.}
We compare the proposed overlap-expansion method with two directed or
direction-aware references. The first is our previous Di-Cosine Overlap
Algorithm, which is included as an internal ablation baseline. This comparison
isolates the effect of replacing the previous directed cosine representation
by the proposed forward--backward Green cosine geometry while keeping the
same two-stage expansion philosophy.

The second reference is CoDA Directed Affiliation, an external model-based
method designed for directed networks. CoDA is included as a supplementary
directed affiliation baseline. In the oracle-initialized experiments, the main
comparison is between the two cosine-based expansion rules, while CoDA is
reported only as an external unsupervised reference.

\begin{itemize}
\item \textbf{Di-Cosine Overlap Algorithm}. This is our previously proposed
cosine-based overlap expansion algorithm from
\cite{DoPhan2025OverlapCosine}. Starting from an initial disjoint partition,
the method assigns a vertex to additional communities according to cosine
similarity between the vertex representation and community representatives.
It is included here as an internal ablation baseline, so that the effect of
replacing the earlier cosine representation by the proposed
forward--backward Green representation can be evaluated separately from the
general two-stage expansion strategy.

    \item \textbf{Di-Green-FB-Cosine Overlap}. This is the proposed overlap
    expansion method. It uses the forward--backward Green cosine geometry
    introduced in Section~\ref{sec:green-fb-overlap}. Starting from an initial
    disjoint partition, the method computes community-adaptive cosine
    thresholds and assigns a vertex to an additional community when its
    forward--backward Green profile is sufficiently aligned with that
    community. In the oracle-initialized experiments, the initial disjoint
    partition is obtained from the ground-truth cover by retaining one
    membership per vertex. In the end-to-end experiments, the initial
    partition is produced algorithmically by
    Di-Green-FB-cosine-KMeans.

    \item \textbf{CoDA Directed Affiliation}. CoDA, or Communities through
    Directed Affiliations, is a model-based method designed for directed
    networks \cite{YangLeskovec2014CoDA}. It represents vertices through
    outgoing and incoming community affiliation strengths and is able to
    capture both cohesive communities and two-mode directed structures. We
    include CoDA as a directed overlapping community detection baseline.

\end{itemize}

All methods are evaluated on the same graph instances and with the same
metrics. Unless explicitly stated otherwise, the parameters reported in the
main tables are fixed before evaluation and are not tuned separately for each
test instance. Ground-truth labels are used only for post-hoc evaluation. When
a parameter grid or sensitivity analysis is reported, it is presented as an
additional diagnostic study rather than as a supervised selection procedure
for the main reported results.

The two cosine-based overlap expansion methods, namely Di-Cosine Overlap and
Di-Green-FB-Cosine Overlap, require an initial disjoint partition. In the main
pipeline of this paper, we use Di-Green-FB-cosine-KMeans as the default
initialization, so that both the disjoint partition and the overlap expansion
are based on the same Green forward--backward geometry.

To isolate the quality of the expansion rule itself, we also consider an
oracle initialization in which each vertex is assigned to exactly one of its
planted communities and all additional overlapping memberships are removed.
This oracle setting is used only as a controlled ablation of the expansion
stage. CoDA is run as an external directed affiliation baseline and does not
use this initial partition.

\subsection{Parameter policy}
\label{subsec:parameter-policy}

For the proposed disjoint method, we use
\[
    \alpha=0.95,\qquad T=8,\qquad \lambda=0.5
\]
unless stated otherwise. Here \(\alpha\) is the teleportation parameter,
\(T\) is the truncation length, and \(\lambda\) balances forward and backward
Green coordinates.

For the proposed overlap method, we use
\[
    \alpha=0.90,\qquad T=10,\qquad \lambda=0.5.
\]
The community-adaptive overlap rule uses the fixed configuration
\[
    q=0.1,\qquad
    \delta=0.05,\qquad
    \eta=0.20,\qquad
    \theta_{\min}=-0.40,\qquad
    \varepsilon=0.
\]
Ground-truth labels are used only for evaluation, not for selecting the
reported output of individual test instances. Sensitivity results over
parameter grids, when reported, are used only to assess robustness of the
method and are not used to choose a separate best parameter setting for each
test graph.

\subsection{Disjoint directed community detection}
\label{subsec:disjoint-experiments}

The disjoint experiments evaluate whether the proposed Green
forward--backward coordinates provide a useful embedding for ordinary
non-overlapping directed community detection. All \(K\)-based methods are
given the planted number of communities. Directed Infomap is included as a
flow-based reference and is allowed to return its own number of communities.

We begin with a small hitting-time and Green-coordinate ablation. This
experiment directly tests the main motivation of the paper: raw hitting-time
rows are natural accessibility descriptors, but they are not well suited for
cosine comparison because of the common stationary baseline. We then evaluate the proposed method on two directed planted benchmark
families. Experiment D1 uses directed heterogeneous Gaussian partition graphs
and tests three regimes: a sparse regime with \(K=8\), a medium-size regime
with \(K=12\), and a large regime with \(K=15\). Experiment D2 uses directed
degree-corrected block graphs and follows the same increasing difficulty
pattern, thereby testing robustness to heterogeneous in-degree and out-degree
propensities. Finally, Experiment D3 evaluates full preprocessed real directed
networks using directed modularity as an internal structural-quality measure.

\subsubsection{Experiment D0: Hitting-time cosine versus centered Green cosine}
\label{subsubsec:exp-disjoint-d0}

This diagnostic experiment tests the main motivation behind the proposed
Green cosine geometry. Hitting times are natural accessibility quantities for
directed graphs. In the hitting-time coordinate construction of
Dang, Do, and Phan \cite{DangDoPhan2023}, a vertex \(i\) is represented by the
stationary-weighted hitting-time row
\[
    C_{\mathrm{HT}}(i)
    =
    H_i\Phi^{1/2}
    =
    \bigl(
    \sqrt{\phi_1}H(i,1),\ldots,\sqrt{\phi_n}H(i,n)
    \bigr),
\]
and the induced distance is
\[
    r_{ij}
    =
    \|H_i\Phi^{1/2}-H_j\Phi^{1/2}\|_2.
\]
This construction is based on the idea that vertices in the same community
should have similar expected hitting times to the rest of the graph.

However, the present diagnostic experiment asks a different question: whether
hitting-time rows are suitable for direct cosine comparison. Let \(Z\) denote
the fundamental matrix of the teleported Markov chain. By the standard
hitting-time identity
\[
    H(i,j)=\frac{Z_{jj}-Z_{ij}}{\phi_j},
\]
or equivalently
\[
    \mathcal G(P)_{ij}
    =
    \phi_j\bigl(H(\phi,j)-H(i,j)\bigr),
\]
one obtains
\[
    H_i
    =
    h-\mathcal G(P)_{i\bullet}\Phi^{-1},
\]
where
\[
    h=(H(\phi,1),\ldots,H(\phi,n))
\]
does not depend on the source vertex \(i\). Thus raw hitting-time rows contain
a common source-independent baseline. Since cosine similarity is not
translation-invariant, applying cosine directly to such rows may distort the
geometry.

Thus \(\mathcal G(P)_{i\bullet}\Phi^{-1}\) is the centered counterpart of the
unweighted hitting-time row \(H_i\). The stationary-weighted variant
\(\mathcal G(P)_{i\bullet}\Phi^{-1/2}\) is not reported separately, since the
purpose of this diagnostic experiment is to test whether removing the common
baseline improves cosine comparison.

We use directed heterogeneous Gaussian partition graphs with
\[
    N\in\{500,1000,1500\},\qquad
    K=8,\qquad
    \bar d=10,\qquad
    \mu\in\{0.10,0.20,0.30\}.
\]
The results are averaged over five random seeds for each parameter combination
\((N,\mu)\). The results are averaged over five random seeds for each parameter
combination. To keep the main tables readable, we report mean values in the
main text and use the same random seeds for all compared methods.

All methods are given the planted number of communities \(K\), and the final
partition is obtained by spherical \(K\)-means on normalized coordinates.
In the reported tables, we compare three coordinate choices:
\[
\begin{aligned}
    &\text{Raw-HT cosine:}
    &&x_i=H_i,\\[1mm]
    &\text{Centered unweighted HT / Green cosine:}
    &&x_i=\mathcal G(P)_{i\bullet}\Phi^{-1},\\[1mm]
    &\text{Truncated diffusive Green-FB cosine:}
    &&x_i=z_i^{(1/2)}
    =
    \left(
    \frac{1}{\sqrt 2}
    \frac{\mathcal G_T^{\mathrm{diff}}(P_\alpha^+)_{i\bullet}}
    {\|\mathcal G_T^{\mathrm{diff}}(P_\alpha^+)_{i\bullet}\|_2},
    \frac{1}{\sqrt 2}
    \frac{\mathcal G_T^{\mathrm{diff}}(P_\alpha^-)_{i\bullet}}
    {\|\mathcal G_T^{\mathrm{diff}}(P_\alpha^-)_{i\bullet}\|_2}
    \right).
\end{aligned}
\]
The first representation tests direct cosine comparison on raw hitting-time
rows. The second removes the common hitting-time baseline from the unweighted
hitting-time row and keeps only the source-dependent centered component. The
third is the practical forward--backward diffusive Green coordinate used in
the proposed method.

\begin{table}[H]
\centering
\caption{Hitting-time cosine versus centered Green cosine. Results are
averaged over \(N\in\{500,1000,1500\}\) and
\(\mu\in\{0.10,0.20,0.30\}\).}
\label{tab:disjoint-d0-overall}
\begin{tabular}{lcccc}
\toprule
Method & NMI & ARI & PairF1 & \(Q_{\mathrm{dir}}\) \\
\midrule
Raw-HT-cosine-KMeans
& 0.5322 & 0.3192 & 0.4614 & 0.3612 \\
\(\mathcal G(P)\Phi^{-1}\)-cosine-KMeans
& 0.9896 & 0.9911 & 0.9923 & 0.6764 \\
\textbf{Truncated Di-Green-FB-cosine-KMeans}
& \textbf{0.9947} & \textbf{0.9954} & \textbf{0.9960}
& \textbf{0.6774} \\
\bottomrule
\end{tabular}
\end{table}

\begin{table}[H]
\centering
\caption{Effect of the mixing parameter on NMI in Experiment D0.}
\label{tab:disjoint-d0-mu}
\begin{tabular}{lccc}
\toprule
Method & \(\mu=0.10\) & \(\mu=0.20\) & \(\mu=0.30\) \\
\midrule
Raw-HT-cosine-KMeans
& 0.8618 & 0.5482 & 0.1865 \\
\(\mathcal G(P)\Phi^{-1}\)-cosine-KMeans
& 0.9993 & 0.9917 & 0.9777 \\
\textbf{Truncated Di-Green-FB-cosine-KMeans}
& \textbf{1.0000} & \textbf{0.9990} & \textbf{0.9852} \\
\bottomrule
\end{tabular}
\end{table}

\paragraph{Main observation.}
Raw hitting-time rows are informative in the easiest regime, but they are not
stable under cosine comparison. Their NMI drops from \(0.8618\) at
\(\mu=0.10\) to \(0.1865\) at \(\mu=0.30\). After removing the common
hitting-time baseline, the centered coordinate
\(\mathcal G(P)_{i\bullet}\Phi^{-1}\)
becomes much more stable, with average NMI \(0.9896\). The proposed truncated
forward--backward Green coordinate gives the best overall performance. This supports the main point of the paper: hitting-time information is useful,
but raw hitting-time rows are not the right objects for cosine comparison; a
centered Green-type representation is needed. The role of
\(\mathcal G(P)\Phi^{-1}\)-cosine in this experiment is diagnostic: it verifies
that removing the hitting-time baseline is crucial. The proposed method itself
uses the truncated forward--backward diffusive Green coordinate, which combines
centering, exclusion of the time-zero self-spike, and incoming--outgoing
diffusion information.

\subsubsection{Experiment D1: Directed Heterogeneous Gaussian Partition Benchmarks}
\label{subsubsec:exp-disjoint-d1}

\paragraph{D1a. Sparse Gaussian Partition Graphs with \(K=8\).}

This experiment uses
\[
    N\in\{1000,1500,2000\},
    \qquad
    K=8,
    \qquad
    \bar d=5,
    \qquad
    \mu\in\{0.10,0.20,0.30\}.
\]
The results are averaged over five random seeds for each parameter combination
\((N,\mu)\). Standard deviations are omitted from the main tables for
readability and are reported in the supplementary material.

\begin{table}[H]
\centering
\caption{Disjoint detection on sparse heterogeneous directed graphs.
Results are averaged over \(N\in\{1000,1500,2000\}\) and
\(\mu\in\{0.10,0.20,0.30\}\).}
\label{tab:disjoint-d1-overall}
\begin{tabular}{lccccc}
\toprule
Method & NMI & ARI & PairF1 & \(Q_{\mathrm{dir}}\) & Avg. returned \(K\) \\
\midrule
\textbf{Di-Green-FB-cosine-KMeans}
& \textbf{0.9297} & \textbf{0.9333} & \textbf{0.9419}
& \textbf{0.6722} & 8.00 \\
rPCA
& 0.8955 & 0.8401 & 0.8621 & 0.6435 & 8.00 \\
oPCA
& 0.6837 & 0.4045 & 0.5096 & 0.4806 & 8.00 \\
D-SCOREq
& 0.6739 & 0.5548 & 0.6143 & 0.4903 & 8.00 \\
Directed Infomap
& 0.6699 & 0.5639 & 0.6236 & 0.4861 & 48.89 \\
D-SCORE
& 0.6393 & 0.5371 & 0.6014 & 0.4980 & 8.00 \\
\bottomrule
\end{tabular}
\end{table}

\begin{table}[H]
\centering
\caption{Effect of the mixing parameter on NMI in Experiment D1.
Results are averaged over \(N\in\{1000,1500,2000\}\).}
\label{tab:disjoint-d1-mu}
\begin{tabular}{lccc}
\toprule
Method & \(\mu=0.10\) & \(\mu=0.20\) & \(\mu=0.30\) \\
\midrule
\textbf{Di-Green-FB-cosine-KMeans}
& \textbf{0.9835} & \textbf{0.9429} & \textbf{0.8626} \\
rPCA
& 0.9833 & 0.9267 & 0.7766 \\
D-SCOREq
& 0.7937 & 0.6684 & 0.5595 \\
oPCA
& 0.7657 & 0.6795 & 0.6060 \\
D-SCORE
& 0.7370 & 0.6299 & 0.5509 \\
Directed Infomap
& 0.6677 & 0.7464 & 0.5957 \\
\bottomrule
\end{tabular}
\end{table}

\paragraph{Main observations.}
The proposed method obtains the best average NMI, ARI, PairF1, and directed
modularity among the methods with complete runs. The strongest competitor is
rPCA. At \(\mu=0.10\), the two methods are nearly tied, while at
\(\mu=0.30\) the proposed method has a clearer advantage. Directed Infomap
tends to return many more than \(K=8\) communities in this sparse setting,
which partly explains its lower ARI and PairF1.

\paragraph{D1b. Medium-Size Gaussian Partition Graphs with \(K=12\).}

This experiment uses
\[
    N\in\{3000,4000,5000\},
    \qquad
    K=12,
    \qquad
    \bar d=10,
    \qquad
    \mu\in\{0.10,0.20,0.30\}.
\]
The results are averaged over five random seeds for each parameter combination
\((N,\mu)\). Standard deviations are omitted from the main tables for
readability and are reported in the supplementary material.

\begin{table}[H]
\centering
\caption{Overall performance in Experiment D1b.}
\label{tab:disjoint-d1b-overall}
\begin{tabular}{lccccc}
\toprule
Method & NMI & ARI & PairF1 & \(Q_{\mathrm{dir}}\) & Avg. returned \(K\) \\
\midrule
\textbf{Di-Green-FB-cosine-KMeans}
& \textbf{0.9308} & \textbf{0.9378} & \textbf{0.9458}
& \textbf{0.6734} & 12.00 \\
rPCA
& 0.8798 & 0.8137 & 0.8394 & 0.6383 & 12.00 \\
oPCA
& 0.6418 & 0.3361 & 0.4593 & 0.4564 & 12.00 \\
D-SCOREq
& 0.5947 & 0.4750 & 0.5454 & 0.4347 & 12.00 \\
D-SCORE
& 0.5808 & 0.4721 & 0.5481 & 0.4636 & 12.00 \\
Directed Infomap
& 0.5321 & 0.3583 & 0.4119 & 0.4082 & 19 \\
\bottomrule
\end{tabular}
\end{table}

\begin{table}[H]
\centering
\caption{Effect of graph size on NMI in Experiment D1b. Results are averaged
over \(\mu\in\{0.10,0.20,0.30\}\).}
\label{tab:disjoint-d1b-size}
\begin{tabular}{lccc}
\toprule
Method & \(N=3000\) & \(N=4000\) & \(N=5000\) \\
\midrule
\textbf{Di-Green-FB-cosine-KMeans}
& \textbf{0.9272} & \textbf{0.9329} & \textbf{0.9324} \\
rPCA
& 0.8937 & 0.8677 & 0.8781 \\
oPCA
& 0.6702 & 0.6269 & 0.6284 \\
D-SCOREq
& 0.6347 & 0.5966 & 0.5528 \\
D-SCORE
& 0.6192 & 0.5722 & 0.5511 \\
Directed Infomap
& 0.5424 & 0.5469 & 0.5069 \\
\bottomrule
\end{tabular}
\end{table}

\begin{table}[H]
\centering
\caption{Effect of the mixing parameter on NMI in Experiment D1b.
Results are averaged over \(N\in\{3000,4000,5000\}\).}
\label{tab:disjoint-d1b-mu}
\begin{tabular}{lccc}
\toprule
Method & \(\mu=0.10\) & \(\mu=0.20\) & \(\mu=0.30\) \\
\midrule
\textbf{Di-Green-FB-cosine-KMeans}
& \textbf{0.9853} & \textbf{0.9494} & \textbf{0.8578} \\
rPCA
& 0.9836 & 0.8789 & 0.7769 \\
oPCA
& 0.7120 & 0.6632 & 0.5501 \\
D-SCOREq
& 0.6847 & 0.6214 & 0.4779 \\
D-SCORE
& 0.6451 & 0.5987 & 0.4987 \\
Directed Infomap
& 0.6743 & 0.5642 & 0.3578 \\
\bottomrule
\end{tabular}
\end{table}

\paragraph{Main observations.}
The proposed method remains the best method on average in this medium-size
regime. Its advantage over rPCA is small at \(\mu=0.10\), but becomes larger
at \(\mu=0.20\) and \(\mu=0.30\). The size-wise table shows that the proposed
method is stable as \(N\) increases from \(3000\) to \(5000\), whereas the
other baselines remain substantially lower except rPCA.

\paragraph{D1c. Large Gaussian Partition Graphs with \(K=15\).}

This experiment uses
\[
    N\in\{8000,9000,10000\},
    \qquad
    K=15,
    \qquad
    \bar d=10,
    \qquad
    \mu\in\{0.10,0.20,0.30\}.
\]
The results are averaged over five random seeds for each parameter combination
\((N,\mu)\). Standard deviations are omitted from the main tables for
readability and are reported in the supplementary material.

\begin{table}[H]
\centering
\caption{Overall performance in Experiment D1c.}
\label{tab:disjoint-d1c-overall}
\begin{tabular}{lccccc}
\toprule
Method & NMI & ARI & PairF1 & \(Q_{\mathrm{dir}}\) & Avg. returned \(K\) \\
\midrule
\textbf{Di-Green-FB-cosine-KMeans}
& \textbf{0.9273} & \textbf{0.9363} & \textbf{0.9446}
& \textbf{0.6732} & 15.00 \\
rPCA
& 0.8770 & 0.8116 & 0.8378 & 0.6373 & 15.00 \\
oPCA
& 0.6442 & 0.3356 & 0.4595 & 0.4629 & 15.00 \\
D-SCOREq
& 0.5815 & 0.4748 & 0.5439 & 0.4288 & 15.00 \\
D-SCORE
& 0.5714 & 0.4575 & 0.5362 & 0.4540 & 15.00 \\
Directed Infomap
& 0.5160 & 0.4126 & 0.4483 & 0.4429 & 267.67 \\
\bottomrule
\end{tabular}
\end{table}

\begin{table}[H]
\centering
\caption{Effect of the mixing parameter in Experiment D1c. Entries are averaged
over \(N\in\{8000,9000,10000\}\).}
\label{tab:disjoint-d1c-mu}
\begin{tabular}{lccc}
\toprule
Method & \(\mu=0.10\) & \(\mu=0.20\) & \(\mu=0.30\) \\
\midrule
\textbf{Di-Green-FB-cosine-KMeans}
& 0.9824 & \textbf{0.9418} & \textbf{0.8576} \\
rPCA
& \textbf{0.9849} & 0.8880 & 0.7582 \\
oPCA
& 0.7113 & 0.6579 & 0.5681 \\
D-SCOREq
& 0.6360 & 0.5980 & 0.5159 \\
D-SCORE
& 0.5938 & 0.5960 & 0.5326 \\
Directed Infomap
& 0.7535 & 0.5187 & 0.2758 \\
\bottomrule
\end{tabular}
\end{table}

\paragraph{Main observations.}
The proposed method has the best overall performance in this larger synthetic
regime. rPCA is slightly better at \(\mu=0.10\), but the proposed method is
clearly better at \(\mu=0.20\) and \(\mu=0.30\). Directed Infomap tends to
overpartition the graph in the difficult regime, as reflected by its large
average \(K_{\mathrm{pred}}\).

\subsubsection{Experiment D2: Directed Degree-Corrected Block Benchmarks}
\label{subsubsec:exp-disjoint-d2}

\paragraph{D2a. Directed Degree-Corrected Block Graphs with \(K=8\).}

This experiment uses directed degree-corrected block graphs with
\[
    N\in\{1000,1500,2000\},
    \qquad
    K=8,
    \qquad
    \bar d=10,
    \qquad
    \mu\in\{0.10,0.20,0.30\}.
\]
The degree-correction parameters are generated independently for outgoing and
incoming roles and normalized within each planted community. The results are
averaged over five random seeds for each parameter combination
\((N,\mu)\). Standard deviations are omitted from the main tables for
readability and are reported in the supplementary material.

\begin{table}[H]
\centering
\caption{Disjoint detection on directed degree-corrected block graphs.
Results are averaged over \(N\in\{1000,1500,2000\}\) and
\(\mu\in\{0.10,0.20,0.30\}\).}
\label{tab:disjoint-d2a-overall}
\begin{tabular}{lccccc}
\toprule
Method & NMI & ARI & PairF1 & \(Q_{\mathrm{dir}}\) & Avg. returned \(K\) \\
\midrule
\textbf{Di-Green-FB-cosine-KMeans}
& \textbf{0.9260} & \textbf{0.9375} & \textbf{0.9456}
& \textbf{0.6621} & 8.00 \\
D-SCOREq
& 0.8907 & 0.9036 & 0.9162 & 0.6481 & 8.00 \\
D-SCORE
& 0.8242 & 0.7980 & 0.8256 & 0.6090 & 8.00 \\
Directed Infomap
& 0.7853 & 0.7700 & 0.7919 & 0.5698 & 47.67 \\
rPCA
& 0.7100 & 0.4678 & 0.5518 & 0.5808 & 8.00 \\
oPCA
& 0.4996 & 0.1426 & 0.3088 & 0.3967 & 8.00 \\
\bottomrule
\end{tabular}
\end{table}

\begin{table}[H]
\centering
\caption{Effect of the mixing parameter on NMI in Experiment D2a.
Results are averaged over \(N\in\{1000,1500,2000\}\).}
\label{tab:disjoint-d2a-mu}
\begin{tabular}{lccc}
\toprule
Method & \(\mu=0.10\) & \(\mu=0.20\) & \(\mu=0.30\) \\
\midrule
\textbf{Di-Green-FB-cosine-KMeans}
& \textbf{0.9784} & \textbf{0.9353} & \textbf{0.8643} \\
D-SCOREq
& 0.9528 & 0.8981 & 0.8214 \\
D-SCORE
& 0.8643 & 0.8477 & 0.7607 \\
Directed Infomap
& 0.9397 & 0.8209 & 0.5953 \\
rPCA
& 0.7266 & 0.7087 & 0.6949 \\
oPCA
& 0.5228 & 0.4933 & 0.4828 \\
\bottomrule
\end{tabular}
\end{table}

\paragraph{Main observations.}
The proposed method obtains the best average NMI, ARI, PairF1, and directed
modularity among the methods with complete runs. The strongest competitor is
D-SCOREq, which is expected because it is designed to reduce the effect of
degree heterogeneity by row-normalizing directed spectral features. However,
Di-Green-FB-cosine-KMeans remains consistently better across all mixing
levels. At \(\mu=0.10\), both methods perform very well, while at
\(\mu=0.30\) the proposed method has a clearer advantage. Directed Infomap is
competitive in the easiest regime but tends to return many more than
\(K=8\) communities, which partly explains its lower ARI and PairF1 on
average.

\paragraph{D2b. Medium-Size Directed Degree-Corrected Block Graphs with \(K=12\).}

This experiment increases both the graph size and the number of planted
communities. It uses directed degree-corrected block graphs with
\[
    N\in\{3000,4000,5000\},
    \qquad
    K=12,
    \qquad
    \bar d=10,
    \qquad
    \mu\in\{0.10,0.20,0.30\}.
\]
As in Experiment D2a, the degree-correction parameters are generated
independently for outgoing and incoming roles and normalized within each
planted community. The results are averaged over five random seeds for each
parameter combination \((N,\mu)\). Standard deviations are omitted from the
main tables for readability and are reported in the supplementary material.

\begin{table}[H]
\centering
\caption{Disjoint detection on medium-size directed degree-corrected block
graphs. Results are averaged over \(N\in\{3000,4000,5000\}\) and
\(\mu\in\{0.10,0.20,0.30\}\).}
\label{tab:disjoint-d2b-overall}
\begin{tabular}{lccccc}
\toprule
Method & NMI & ARI & PairF1 & \(Q_{\mathrm{dir}}\) & Avg. returned \(K\) \\
\midrule
\textbf{Di-Green-FB-cosine-KMeans}
& \textbf{0.9186} & \textbf{0.9189} & \textbf{0.9263}
& \textbf{0.6920} & 12.00 \\
D-SCOREq
& 0.8955 & 0.9061 & 0.9146 & 0.6866 & 12.00 \\
D-SCORE
& 0.8088 & 0.7252 & 0.7532 & 0.6259 & 12.00 \\
Directed Infomap
& 0.7718 & 0.7189 & 0.7340 & 0.5696 & 120.78 \\
rPCA
& 0.6907 & 0.3427 & 0.4237 & 0.5862 & 12.00 \\
oPCA
& 0.4688 & 0.0716 & 0.2100 & 0.3665 & 12.00 \\
\bottomrule
\end{tabular}
\end{table}

\begin{table}[H]
\centering
\caption{Effect of the mixing parameter on NMI in Experiment D2b.
Results are averaged over \(N\in\{3000,4000,5000\}\).}
\label{tab:disjoint-d2b-mu}
\begin{tabular}{lccc}
\toprule
Method & \(\mu=0.10\) & \(\mu=0.20\) & \(\mu=0.30\) \\
\midrule
Di-Green-FB-cosine-KMeans
& 0.9486 & \textbf{0.9350} & \textbf{0.8723} \\
D-SCOREq
& \textbf{0.9547} & 0.9007 & 0.8313 \\
D-SCORE
& 0.8870 & 0.8035 & 0.7360 \\
Directed Infomap
& 0.9528 & 0.8294 & 0.5332 \\
rPCA
& 0.7116 & 0.6910 & 0.6694 \\
oPCA
& 0.4849 & 0.4719 & 0.4497 \\
\bottomrule
\end{tabular}
\end{table}

\paragraph{Main observations.}
The proposed method obtains the best average NMI, ARI, PairF1, and directed
modularity in this medium-size degree-corrected setting. D-SCOREq is again
the strongest \(K\)-based competitor and is slightly better at
\(\mu=0.10\). However, as the mixing parameter increases, the proposed method
becomes more stable: it outperforms D-SCOREq at both \(\mu=0.20\) and
\(\mu=0.30\). Directed Infomap is highly competitive in the easiest regime,
but it strongly overpartitions the graph as the mixing level increases,
returning on average far more than the planted \(K=12\) communities. This
leads to a clear degradation in ARI, PairF1, and directed modularity in the
harder regimes.

\paragraph{D2c. Large Directed Degree-Corrected Block Graphs with \(K=15\).}

This experiment further increases the graph size and the number of planted
communities. It uses directed degree-corrected block graphs with
\[
    N\in\{8000,9000,10000\},
    \qquad
    K=15,
    \qquad
    \bar d=10,
    \qquad
    \mu\in\{0.10,0.20,0.30\}.
\]
The degree-correction parameters are generated independently for outgoing and
incoming roles and normalized within each planted community. The results are
averaged over five random seeds for each parameter combination
\((N,\mu)\). Standard deviations are omitted from the main tables for
readability and are reported in the supplementary material.

\begin{table}[H]
\centering
\caption{Disjoint detection on large directed degree-corrected block graphs.
Results are averaged over \(N\in\{8000,9000,10000\}\) and
\(\mu\in\{0.10,0.20,0.30\}\).}
\label{tab:disjoint-d2c-overall}
\begin{tabular}{lccccc}
\toprule
Method & NMI & ARI & PairF1 & \(Q_{\mathrm{dir}}\) & Avg. returned \(K\) \\
\midrule
\textbf{Di-Green-FB-cosine-KMeans}
& \textbf{0.9135} & \textbf{0.9094} & \textbf{0.9164}
& \textbf{0.7028} & 15.00 \\
D-SCOREq
& 0.8935 & 0.9056 & 0.9129 & 0.7010 & 15.00 \\
D-SCORE
& 0.8252 & 0.7666 & 0.7862 & 0.6509 & 15.00 \\
rPCA
& 0.6927 & 0.3072 & 0.3817 & 0.5922 & 15.00 \\
oPCA
& 0.4643 & 0.0564 & 0.1796 & 0.3617 & 15.00 \\
Directed Infomap
& 0.4307 & 0.1523 & 0.2296 & 0.3195 & 240.56 \\
\bottomrule
\end{tabular}
\end{table}

\begin{table}[H]
\centering
\caption{Effect of the mixing parameter on NMI in Experiment D2c.
Results are averaged over \(N\in\{8000,9000,10000\}\).}
\label{tab:disjoint-d2c-mu}
\begin{tabular}{lccc}
\toprule
Method & \(\mu=0.10\) & \(\mu=0.20\) & \(\mu=0.30\) \\
\midrule
\textbf{Di-Green-FB-cosine-KMeans}
& \textbf{0.9657} & \textbf{0.9174} & \textbf{0.8575} \\
D-SCOREq
& 0.9545 & 0.9005 & 0.8256 \\
D-SCORE
& 0.8972 & 0.8345 & 0.7439 \\
rPCA
& 0.7082 & 0.6947 & 0.6751 \\
oPCA
& 0.4779 & 0.4706 & 0.4444 \\
Directed Infomap
& 0.4116 & 0.4696 & 0.4110 \\
\bottomrule
\end{tabular}
\end{table}

\paragraph{Main observations.}
The proposed method remains the best method on average in this large
degree-corrected regime. It obtains the highest average NMI, ARI, PairF1, and
directed modularity, while always returning the prescribed number
\(K=15\) of communities. D-SCOREq is again the strongest \(K\)-based
competitor and is very close in ARI, PairF1, and \(Q_{\mathrm{dir}}\).
However, Di-Green-FB-cosine-KMeans has a consistent advantage in NMI across
all three mixing levels. The gap is small at \(\mu=0.10\), but becomes clearer
as the mixing parameter increases.

Directed Infomap behaves differently in this large sparse degree-corrected
setting. It substantially overpartitions the graph, returning on average
about \(240\) communities instead of the planted \(K=15\). This overpartitioning
leads to much lower ARI, PairF1, and directed modularity. The results suggest
that, even when the number of vertices reaches \(10^4\) and the communities
are affected by both in-degree and out-degree heterogeneity, the proposed
forward--backward Green geometry remains stable and competitive.

\subsubsection{Experiment D3: Full real directed networks}
\label{subsubsec:exp-disjoint-d3}

The previous experiments use synthetic directed benchmarks with known planted
communities. We now complement them with an evaluation on real directed
networks. Since most real networks do not provide a reliable disjoint
ground-truth partition compatible with all methods, we do not report NMI, ARI,
or PairF1 in this experiment. Instead, we compare the structural quality of
the detected partitions using directed modularity \(Q_{\mathrm{dir}}\).

The purpose of this experiment is therefore different from the synthetic
experiments. It does not claim ground-truth recovery. Rather, it tests whether
the proposed forward--backward Green geometry produces structurally coherent
directed partitions on real networks, compared with directed spectral
embeddings and a random-walk flow method.

We emphasize that this experiment is not intended as a competition against
algorithms that directly optimize directed modularity, such as directed
Louvain or directed Leiden variants. Since the evaluation metric in this
section is \(Q_{\mathrm{dir}}\), a modularity optimizer would be a natural
additional reference. The purpose here is more limited: to check whether a
partition obtained from the proposed Green diffusion geometry, without
directly optimizing \(Q_{\mathrm{dir}}\), nevertheless attains competitive
directed modularity on diverse real directed networks.

We use the full preprocessed directed graphs, without node sampling and
without extracting a largest connected component. For temporal or repeated-edge data, we convert
the graph to an unweighted directed graph by collapsing repeated directed
interactions into a single binary edge. For signed networks, we use the
positive directed subgraph, since the Markov transition matrix and
\(Q_{\mathrm{dir}}\) used in this paper require a nonnegative adjacency
matrix. Thus, \(m_{\mathrm{raw}}\) denotes the number of raw records in the
original file, while \(m\) denotes the number of binary directed edges used in
the experiment. Similarly, for the positive signed networks, \(n\) denotes the
number of vertices incident to at least one retained positive directed edge.

For every \(K\)-based embedding method, we sweep the same grid
\[
K\in\{4,6,8,10,12,16,24,32\}
\]
and report the best directed modularity over this grid. Thus the model
selection protocol is identical for oPCA, rPCA, D-SCORE, D-SCOREq, and
Di-Green-FB-cosine-KMeans. Directed Infomap is not given \(K\); it
automatically returns its own number of communities. The same directed
modularity function is then used to evaluate all outputs. For compactness, the
selected \(K\) is reported only for the proposed method in the main table.

\begin{table}[H]
\centering
\caption{Real directed networks used in Experiment D3. The reported \(n\) and
\(m\) are the numbers of vertices and binary directed edges used after
preprocessing.}
\label{tab:real-datasets-d3}
\renewcommand{\arraystretch}{1.18}
\resizebox{\textwidth}{!}{
\begin{tabular}{
|>{\raggedright\arraybackslash}p{3.6cm}
|>{\raggedright\arraybackslash}p{2.6cm}
|r
|r
|>{\raggedright\arraybackslash}p{8.0cm}|
}
\hline
\textbf{Dataset} & \textbf{Type} & \(\boldsymbol{n}\) & \(\boldsymbol{m}\) & \textbf{Description} \\
\hline

email-Eu-core
& email
& 986 & 24,929
& Email communication network from a large European research institution.
A directed edge \(u\to v\) means that person \(u\) sent at least one email to
person \(v\). The original dataset also contains department labels, but these
labels are not used in this modularity-only experiment
\cite{SNAPEmailEuCore,YinBensonLeskovecGleich2017}. \\
\hline

CollegeMsg
& social-message
& 1,899 & 20,296
& Temporal private-message network from an online social network at the
University of California, Irvine. A raw record \((u,v,t)\) means that user
\(u\) sent a message to user \(v\) at time \(t\). We collapse all repeated
messages into a static binary directed graph
\cite{SNAPCollegeMsg,PanzarasaOpsahlCarley2009}. \\
\hline

wiki-Vote
& voting
& 7,115 & 103,689
& Wikipedia adminship voting network up to January 2008. Nodes are Wikipedia
users, and a directed edge \(u\to v\) means that user \(u\) voted on user
\(v\)'s adminship request. This is a directed voting network with strong
asymmetry and heterogeneous roles
\cite{SNAPWikiVote} \\
\hline

p2p-Gnutella04
& p2p-snapshot
& 10,876 & 39,994
& Snapshot of the Gnutella peer-to-peer file-sharing network on August 4,
2002. Nodes are hosts in the Gnutella topology, and directed edges represent
connections between hosts. We use this snapshot as one representative
technological directed network
\cite{SNAPGnutella04,SNAPGnutella08,RipeanuFosterIamnitchi2002}. \\
\hline

p2p-Gnutella08
& p2p-snapshot
& 6,301 & 20,777
& Snapshot of the same Gnutella peer-to-peer system on August 8, 2002. We keep
only two Gnutella snapshots in the final table, rather than many consecutive
days, to avoid over-weighting one highly similar data family
\cite{LeskovecKleinbergFaloutsos2007,RipeanuFosterIamnitchi2002}. \\
\hline

soc-sign-bitcoinalpha-pos
& bitcoin-trust
& 3,683 & 22,650
& Directed signed trust network from the Bitcoin Alpha trading platform.
The original edge weight ranges from \(-10\) to \(+10\). We retain only
positive trust ratings and convert them to binary directed edges, producing a
nonnegative graph suitable for Markov diffusion and directed modularity
\cite{KumarSpezzanoSubrahmanianFaloutsos2016}. \\
\hline

soc-sign-bitcoinotc-pos
& bitcoin-trust
& 5,573 & 32,029
& Directed signed trust network from the Bitcoin OTC trading platform.
As in Bitcoin Alpha, an edge records a user rating another user after a
transaction. We use the positive trust subgraph and binarize the retained
directed edges
\cite{KumarSpezzanoSubrahmanianFaloutsos2016}. \\
\hline

wiki-RfA-pos
& wikipedia-admin-vote
& 10,015 & 139,741
& Wikipedia Requests for Adminship dataset from 2003 to May 2013. The original
data contain support, neutral, and oppose votes, together with textual
comments. We retain only positive support votes to obtain a nonnegative
directed voting graph
\cite{WestPaskovLeskovecPotts2014}. \\
\hline

\end{tabular}
}
\end{table}

\begin{table}[H]
\centering
\caption{Real directed network experiment. For \(K\)-based methods, each entry
is the best \(Q_{\mathrm{dir}}\) over
\(K\in\{4,6,8,10,12,16,24,32\}\). Directed Infomap automatically determines
its own number of communities.}
\label{tab:real-d3-qdir}
\resizebox{\textwidth}{!}{
\begin{tabular}{l c c c c c c c}
\toprule
Dataset
& oPCA
& rPCA
& D-SCORE
& D-SCOREq
& Directed Infomap
& \textbf{Di-Green-FB}
& Best \(K\) of Di-Green-FB \\
\midrule

CollegeMsg
& 0.004989
& 0.039618
& 0.039266
& 0.023548
& 0.075373
& \textbf{0.218961}
& 4 \\

email-Eu-core
& 0.303777
& 0.342633
& 0.329602
& 0.395349
& \textbf{0.415180}
& 0.401660
& 8 \\

p2p-Gnutella04
& 0.039609
& 0.051285
& 0.012374
& 0.006079
& 0.256302
& \textbf{0.261930}
& 32 \\

p2p-Gnutella08
& 0.146861
& 0.124313
& 0.014479
& 0.004016
& 0.331302
& \textbf{0.340441}
& 16 \\

soc-sign-bitcoinalpha-pos
& 0.106985
& 0.218919
& 0.331342
& 0.359611
& 0.409717
& \textbf{0.469817}
& 8 \\

soc-sign-bitcoinotc-pos
& 0.108767
& 0.212601
& 0.282383
& 0.383137
& 0.009821
& \textbf{0.474191}
& 10 \\

wiki-RfA-pos
& 0.261137
& 0.310469
& 0.287249
& \textbf{0.448840}
& 0.432033
& 0.443786
& 4 \\

wiki-Vote
& 0.206345
& 0.263301
& 0.215829
& 0.338127
& 0.000999
& \textbf{0.370248}
& 4 \\

\midrule
Average
& 0.1473
& 0.1954
& 0.1891
& 0.2448
& 0.2413
& \textbf{0.3726}
& -- \\

\bottomrule
\end{tabular}
}
\end{table}

\begin{table}[H]
\centering
\caption{Winner by dataset in Experiment D3.}
\label{tab:real-d3-winners}
\begin{tabular}{l l c c}
\toprule
Dataset & Winner & Best \(Q_{\mathrm{dir}}\) & Selected/Predicted \(K\) \\
\midrule
CollegeMsg & Di-Green-FB-cosine-KMeans & 0.218961 & 4 \\
email-Eu-core & Directed Infomap & 0.415180 & 18 \\
p2p-Gnutella04 & Di-Green-FB-cosine-KMeans & 0.261930 & 32 \\
p2p-Gnutella08 & Di-Green-FB-cosine-KMeans & 0.340441 & 16 \\
soc-sign-bitcoinalpha-pos & Di-Green-FB-cosine-KMeans & 0.469817 & 8 \\
soc-sign-bitcoinotc-pos & Di-Green-FB-cosine-KMeans & 0.474191 & 10 \\
wiki-RfA-pos & D-SCOREq & 0.448840 & 4 \\
wiki-Vote & Di-Green-FB-cosine-KMeans & 0.370248 & 4 \\
\bottomrule
\end{tabular}
\end{table}

\paragraph{Main observations.}
Among the embedding-based and flow-based references considered in this
experiment, the proposed Di-Green-FB-cosine-KMeans method obtains the highest
directed modularity on six of the eight real directed networks. On the two
remaining networks, it remains close to the best method: on email-Eu-core,
Directed Infomap obtains \(Q_{\mathrm{dir}}=0.415180\), while Di-Green-FB
obtains \(0.401660\); on wiki-RfA-pos, D-SCOREq obtains
\(0.448840\), while Di-Green-FB obtains \(0.443786\).

These results should be interpreted as evidence that the proposed Green
forward--backward geometry can produce structurally coherent directed
partitions. They should not be interpreted as showing superiority over
dedicated directed-modularity maximization algorithms, which are not included
in this comparison.

The advantage is especially clear on the trust and social-message networks.
On CollegeMsg, the proposed method reaches \(Q_{\mathrm{dir}}=0.218961\),
whereas the strongest non-Green competitor, Directed Infomap, reaches only
\(0.075373\). On the two Bitcoin trust networks, the proposed method obtains
\(0.469817\) and \(0.474191\), clearly above the directed spectral baselines
and Infomap. This suggests that the forward--backward Green geometry is
particularly effective when community structure is encoded through directed
accessibility and asymmetric trust relations.

The Gnutella results are also informative. Directed Infomap is competitive on
the two peer-to-peer snapshots, but Di-Green-FB remains slightly better:
\(0.261930\) versus \(0.256302\) on p2p-Gnutella04, and \(0.340441\) versus
\(0.331302\) on p2p-Gnutella08. Since these snapshots come from the same
Gnutella data family, we include only two representative days in the final
experiment to avoid artificially inflating the number of wins by repeatedly
testing nearly identical network instances.

Averaged over the eight real networks, Di-Green-FB obtains
\[
    \overline Q_{\mathrm{dir}}=0.3726,
\]
compared with \(0.2448\) for D-SCOREq, \(0.2413\) for Directed Infomap,
\(0.1954\) for rPCA, \(0.1891\) for D-SCORE, and \(0.1473\) for oPCA.
This suggests that the proposed Green forward--backward representation is
capable of producing high-\(Q_{\mathrm{dir}}\) partitions on several real
directed networks.

\paragraph{Interpretation.}
Overall, the real-network experiment provides supporting evidence that the
forward--backward Green representation can produce high-\(Q_{\mathrm{dir}}\)
partitions on diverse directed networks. This evidence is structural rather
than ground-truth-based, and should be interpreted together with the synthetic
benchmark results.

\subsection{Overlapping directed community detection}
\label{subsec:overlap-experiments}

The overlapping experiments evaluate whether the proposed Green cosine
geometry can expand a disjoint partition into an overlapping cover. We
consider two types of initialization.

First, oracle initialization assigns each vertex to exactly one of its true
memberships and removes all additional memberships. This setting is used as an
input protocol for the two cosine-based expansion methods and isolates the
overlap expansion rule from the quality of the initial disjoint partition.
Second, end-to-end initialization obtains the initial partition algorithmically
from the graph. This setting evaluates the full pipeline.

The oracle experiments should be interpreted as controlled ablations of the
second stage. The end-to-end experiment is the main practical evaluation of
the complete overlapping community detection pipeline.

\subsubsection{Experiment O1: Oracle initialization on small and medium graphs}
\label{subsubsec:exp-overlap-o1}

This experiment uses oracle disjoint initialization for the two cosine-based
expansion methods and tests the overlap expansion stage. The oracle
initialization is used as an input protocol rather than as a method reported
in the table. CoDA is run independently as an external directed affiliation
baseline.

The comparison with CoDA in this oracle-initialized setting should be
interpreted cautiously, because the two cosine-based expansion methods are
given an oracle disjoint initialization, whereas CoDA is run as a fully
unsupervised affiliation model. Therefore, the primary purpose of this
experiment is to compare Di-Green-FB-Cosine Overlap with Di-Cosine Overlap
under the same initialization. CoDA is included only as a supplementary
external reference.

The graph parameters are
\[
    N\in\{500,1000,1500\},
    \qquad
    K=8,
    \qquad
    \bar d=10,
    \qquad
    o_n=0.15N,
    \qquad
    o_m=2,
\]
with
\[
    \mu\in\{0.10,0.20,0.30\}.
\]
Results are averaged over five random seeds for each parameter combination
\((N,\mu)\), giving \(45\) graph instances. Standard deviations are omitted
from the main tables for readability and are reported in the supplementary
material.

\begin{table}[H]
\centering
\caption{Overlap detection in Experiment O1. The two cosine-based expansion
methods use oracle disjoint initialization; CoDA is run independently as an
external directed affiliation baseline.}
\label{tab:overlap-o1-overall}
\begin{tabular}{lcccc}
\toprule
Method & ONMI & PairF1 & OverlapF1 & Score \\
\midrule
Di-Cosine Overlap Algorithm
& 0.9294 & 0.9628 & 0.8311 & 0.9078 \\
\textbf{Di-Green-FB-Cosine Overlap}
& \textbf{0.9478} & \textbf{0.9747} & \textbf{0.8828}
& \textbf{0.9351} \\
CoDA Directed Affiliation
& 0.2605 & 0.4534 & 0.2935 & 0.3358 \\
\bottomrule
\end{tabular}
\end{table}

\begin{table}[H]
\centering
\caption{Effect of the mixing parameter on overlap Score in Experiment O1.}
\label{tab:overlap-o1-mu-score}
\begin{tabular}{lccc}
\toprule
Method
& \(\mu=0.10\) & \(\mu=0.20\) & \(\mu=0.30\) \\
\midrule
Di-Cosine Overlap Algorithm
& 0.9662 & 0.9264 & 0.8307 \\
\textbf{Di-Green-FB-Cosine Overlap}
& \textbf{0.9789} & \textbf{0.9388} & \textbf{0.8877} \\
CoDA Directed Affiliation
& 0.4558 & 0.3189 & 0.2327 \\
\bottomrule
\end{tabular}
\end{table}

\paragraph{Main observations.}
The proposed Green overlap rule gives the best average ONMI, PairF1,
OverlapF1, and Score. The improvement over Di-Cosine is most visible in the
hardest regime \(\mu=0.30\), where the Green forward--backward geometry gives
a larger gain in OverlapF1. Since this experiment uses oracle initialization,
the result should be interpreted as evidence for the quality of the expansion
rule rather than the full pipeline.

\subsubsection{Experiment O2: Oracle initialization on larger graphs}
\label{subsubsec:exp-overlap-o2}

This experiment again uses oracle disjoint initialization, but increases the
graph size and number of communities. As in Experiment O1, the oracle
initialization is used only to isolate the quality of the overlap expansion
rule. The comparison with CoDA should again be interpreted as supplementary,
because CoDA is run without oracle initialization.

The graph parameters are
\[
    N\in\{3000,4000,5000\},
    \qquad
    K=12,
    \qquad
    \bar d=10,
    \qquad
    o_n=0.15N,
    \qquad
    o_m=2,
\]
with
\[
    \mu\in\{0.10,0.20,0.30\}.
\]
Results are averaged over five random seeds for each parameter combination
\((N,\mu)\), giving \(45\) graph instances. Standard deviations are omitted
from the main tables for readability and are reported in the supplementary
material.

\begin{table}[H]
\centering
\caption{Overlap detection in Experiment O2. The two cosine-based expansion
methods use oracle disjoint initialization; CoDA is run independently as an
external directed affiliation baseline.}
\label{tab:overlap-o2-overall}
\begin{tabular}{lcccc}
\toprule
Method & ONMI & PairF1 & OverlapF1 & Score \\
\midrule
Di-Cosine Overlap Algorithm
& 0.9457 & 0.9707 & 0.8724 & 0.9296 \\
\textbf{Di-Green-FB-Cosine Overlap}
& \textbf{0.9553} & \textbf{0.9762} & \textbf{0.8904}
& \textbf{0.9406} \\
CoDA Directed Affiliation
& 0.1413 & 0.3098 & 0.2699 & 0.2403 \\
\bottomrule
\end{tabular}
\end{table}

\paragraph{Main observations.}
The proposed method remains the best overlap expansion method in the larger
oracle-initialized setting. The gap over Di-Cosine is moderate on average but
becomes clearer in the high-mixing regime, especially for OverlapF1. CoDA performs worse in this particular planted benchmark. Since the two
cosine-based methods use oracle initialization in this experiment, the main
conclusion is that the proposed Green representation improves the expansion
rule relative to the previous Di-Cosine baseline under the same initialization.

\subsubsection{Experiment O3: End-to-end overlapping community detection}
\label{subsubsec:exp-overlap-o3}

This experiment evaluates the full overlapping community detection pipeline.
Unlike Experiments O1 and O2, the initial disjoint partition is not supplied by
the ground truth. For the proposed pipeline, the initial partition is produced
by Di-Green-FB-cosine-KMeans and then expanded by the
Di-Green-FB-Cosine Overlap rule.

The benchmark parameters are
\[
    (N,K)\in\{(1000,8),(3000,12),(5000,15),(7000,20)\},
\]
with
\[
    \bar d=10,
    \qquad
    o_n=0.15N,
    \qquad
    o_m=2,
    \qquad
    \mu\in\{0.10,0.20\}.
\]
This gives eight end-to-end graph instances.

\begin{table}[H]
\centering
\caption{Overall performance in the end-to-end overlap experiment.}
\label{tab:overlap-o3-overall}
\begin{tabular}{lcccc}
\toprule
Method & ONMI & PairF1 & OverlapF1 & Score \\
\midrule
Di-Cosine Overlap Algorithm
& 0.9334 & 0.9651 & 0.8783 & 0.9256 \\
\textbf{Di-Green-FB-Cosine Overlap}
& \textbf{0.9491} & \textbf{0.9731} & \textbf{0.9167}
& \textbf{0.9463} \\
CoDA Directed Affiliation
& 0.2677 & 0.4122 & 0.2894 & 0.3231 \\
\bottomrule
\end{tabular}
\end{table}

\begin{table}[H]
\centering
\caption{Effect of the mixing parameter in the end-to-end overlap experiment.}
\label{tab:overlap-o3-mu-score}
\begin{tabular}{lcc}
\toprule
Method & \(\mu=0.10\) & \(\mu=0.20\) \\
\midrule
Di-Cosine Overlap Algorithm
& 0.9511 & 0.9001 \\
\textbf{Di-Green-FB-Cosine Overlap}
& \textbf{0.9742} & \textbf{0.9184} \\
CoDA Directed Affiliation
& 0.3698 & 0.2764 \\
\bottomrule
\end{tabular}
\end{table}

\begin{table}[H]
\centering
\caption{Detailed performance of the proposed end-to-end pipeline.}
\label{tab:overlap-o3-green-grid}
\begin{tabular}{c c c c c c c}
\toprule
\(N\) & \(K\) & \(\mu\) & ONMI & PairF1 & OverlapF1 & Score \\
\midrule
1000 & 8  & 0.10 & 0.9879 & 0.9958 & 0.9800 & 0.9879 \\
1000 & 8  & 0.20 & 0.9329 & 0.9591 & 0.8757 & 0.9226 \\
3000 & 12 & 0.10 & 0.9745 & 0.9887 & 0.9642 & 0.9758 \\
3000 & 12 & 0.20 & 0.8842 & 0.9328 & 0.7915 & 0.8695 \\
5000 & 15 & 0.10 & 0.9678 & 0.9855 & 0.9521 & 0.9685 \\
5000 & 15 & 0.20 & 0.9492 & 0.9758 & 0.9241 & 0.9497 \\
7000 & 20 & 0.10 & 0.9630 & 0.9813 & 0.9494 & 0.9646 \\
7000 & 20 & 0.20 & 0.9336 & 0.9657 & 0.8968 & 0.9320 \\
\bottomrule
\end{tabular}
\end{table}

\paragraph{Main observations.}
This is the main practical overlap experiment, because the initialization is
computed from the graph rather than supplied by the ground truth. The proposed end-to-end pipeline obtains the best average Score among the
tested overlap methods. The improvement over the previous Di-Cosine expansion
baseline indicates that the Green forward--backward geometry provides a
stronger basis for adding overlapping memberships after the initial disjoint
partition has been computed. The hardest instance is \(N=3000,K=12,\mu=0.20\), where the initial
partition is weaker; nevertheless, the final overlap cover remains
competitive.

\subsection{Summary of experimental findings}
\label{subsec:experiment-summary}

The experiments support the main geometric motivation of the paper. Raw
hitting-time rows are unstable under direct cosine comparison, while centered
Green representations are substantially more robust. On synthetic disjoint
benchmarks, the proposed forward--backward Green coordinate is competitive
with, and often improves over, the tested directed spectral and flow-based
baselines, especially in moderate and high-mixing regimes. The overlap
experiments show that the same geometry improves the expansion step relative
to the previous Di-Cosine baseline, both under oracle initialization and in
the end-to-end setting. The real-network modularity experiment provides
additional structural evidence, but should be interpreted as an internal
quality evaluation rather than as a ground-truth recovery test.

\section{Conclusion}
\label{sec:conclusion}

In this paper, we proposed a forward--backward Green cosine framework for
community detection in directed graphs, with an extension to overlapping
memberships. The starting point was the observation that raw hitting-time
profiles contain a source-independent stationary baseline and are therefore
not well suited for direct cosine comparison. By replacing raw hitting-time
profiles with centered Green diffusion profiles, the proposed representation
removes this background and compares vertices through their source-dependent
deviations from stationarity.

To account for edge asymmetry, we combined the diffusive Green profile of the
forward walk with the corresponding profile on the edge-reversed graph. The
resulting forward--backward coordinate captures both outgoing and incoming
diffusion roles. This geometry was used in two stages: first to construct an
initial disjoint partition through Di-Green-FB-cosine-KMeans, and then to
expand that partition into an overlapping cover by a community-adaptive cosine
rule.

We established basic theoretical support for the proposed geometry, including
the centered-inverse interpretation of the Green operator, the baseline
distortion of raw hitting-time cosine, a truncation bound for finite-time
diffusive Green coordinates, the positive semidefinite kernel property of the
forward--backward cosine, and deterministic margin conditions for the
idealized disjoint and overlapping cosine assignment rules.

Experiments on synthetic directed benchmarks indicate that the proposed
representation is substantially more stable than raw hitting-time cosine and
is competitive with directed spectral and flow-based baselines. The overlap
experiments further suggest that the same Green cosine geometry provides a
useful mechanism for adding secondary memberships when the initial disjoint
partition is reliable. Real-network experiments, evaluated by directed
modularity, provide supporting evidence that the method can produce coherent
directed partitions beyond planted benchmarks.

The main limitation of the present implementation is computational. Full
Green-coordinate construction requires \(O(n^2)\) memory, and the direct
overlap expansion stage can be expensive for very large graphs. Future work
will therefore focus on low-rank, landmark, sketched, and nearest-neighbor
approximations, automatic estimation of the number of communities, and broader
end-to-end validation on real directed networks with reliable overlapping
ground truth.

\bibliographystyle{plain}
\bibliography{references}

@article{Fortunato2010,
  author  = {Fortunato, Santo},
  title   = {Community Detection in Graphs},
  journal = {Physics Reports},
  volume  = {486},
  number  = {3--5},
  pages   = {75--174},
  year    = {2010}
}

@article{MalliarosVazirgiannis2013,
  author  = {Malliaros, Fragkiskos D. and Vazirgiannis, Michalis},
  title   = {Clustering and Community Detection in Directed Networks: A Survey},
  journal = {Physics Reports},
  volume  = {533},
  number  = {4},
  pages   = {95--142},
  year    = {2013}
}

@article{Chung2005,
  author  = {Chung, Fan},
  title   = {Laplacians and the Cheeger Inequality for Directed Graphs},
  journal = {Annals of Combinatorics},
  volume  = {9},
  number  = {1},
  pages   = {1--19},
  year    = {2005}
}

@article{LiZhang2012,
  author  = {Li, Yanhua and Zhang, Zhi-Li},
  title   = {Digraph Laplacian and the Degree of Asymmetry},
  journal = {Internet Mathematics},
  volume  = {8},
  number  = {4},
  pages   = {381--401},
  year    = {2012}
}

@article{Peel2017,
  author  = {Peel, Leto and Larremore, Daniel B. and Clauset, Aaron},
  title   = {The Ground Truth about Metadata and Community Detection in Networks},
  journal = {Science Advances},
  volume  = {3},
  number  = {5},
  pages   = {e1602548},
  year    = {2017}
}

@article{Jebabli2018,
  author  = {Jebabli, Malek and Cherifi, Hocine and Cherifi, Chantal and Hamouda, Atef},
  title   = {Community Detection Algorithm Evaluation with Ground-Truth Data},
  journal = {Physica A: Statistical Mechanics and its Applications},
  volume  = {492},
  pages   = {651--706},
  year    = {2018},
  doi     = {10.1016/j.physa.2017.10.018}
}

@article{Harenberg2014,
  author  = {Harenberg, Steve and Bello, Gonzalo and Gjeltema, Lucas and Ranshous, Stephen and Harlalka, Jitendra and Seay, Ramona and Padmanabhan, Kanchana and Samatova, Nagiza},
  title   = {Community Detection in Large-Scale Networks: A Survey and Empirical Evaluation},
  journal = {WIREs Computational Statistics},
  volume  = {6},
  number  = {6},
  pages   = {426--439},
  year    = {2014},
  doi     = {10.1002/wics.1319}
}

@article{Ponomarenko2021,
  author  = {Ponomarenko, A. and Pitsoulis, L. and Shamshetdinov, M.},
  title   = {Overlapping Community Detection in Networks Based on Link Partitioning and Partitioning Around Medoids},
  journal = {PLOS ONE},
  volume  = {16},
  number  = {8},
  pages   = {e0255717},
  year    = {2021}
}

@article{Chen2010,
  author  = {Chen, Duanbing and Shang, Mingsheng and Lv, Zehua and Fu, Yan},
  title   = {Detecting Overlapping Communities of Weighted Networks via a Local Algorithm},
  journal = {Physica A: Statistical Mechanics and its Applications},
  volume  = {389},
  number  = {19},
  pages   = {4177--4187},
  year    = {2010},
  doi     = {10.1016/j.physa.2010.05.046}
}

@article{Hunter1982,
  author  = {Hunter, Jeffrey J.},
  title   = {Generalized Inverses and Their Application to Applied Probability Problems},
  journal = {Linear Algebra and its Applications},
  volume  = {45},
  pages   = {157--198},
  year    = {1982},
  doi     = {10.1016/0024-3795(82)90218-X}
}

@book{KemenySnell1976,
  author    = {Kemeny, John G. and Snell, J. Laurie},
  title     = {Finite Markov Chains: With a New Appendix ``Generalization of a Fundamental Matrix''},
  series    = {Undergraduate Texts in Mathematics},
  publisher = {Springer},
  address   = {New York},
  year      = {1976},
  isbn      = {978-0-387-90192-3}
}

@techreport{Page1999,
  author      = {Page, Lawrence and Brin, Sergey and Motwani, Rajeev and Winograd, Terry},
  title       = {The {P}age{R}ank Citation Ranking: Bringing Order to the Web},
  institution = {Stanford InfoLab},
  type        = {Technical Report},
  number      = {1999-66},
  year        = {1999},
  note        = {Stanford InfoLab publication 422}
}

@inproceedings{ArthurVassilvitskii2007,
  author    = {Arthur, David and Vassilvitskii, Sergei},
  title     = {$k$-Means++: The Advantages of Careful Seeding},
  booktitle = {Proceedings of the Eighteenth Annual ACM-SIAM Symposium on Discrete Algorithms},
  pages     = {1027--1035},
  publisher = {Society for Industrial and Applied Mathematics},
  address   = {Philadelphia, PA},
  year      = {2007}
}

@article{DhillonModha2001,
  author  = {Dhillon, Inderjit S. and Modha, Dharmendra S.},
  title   = {Concept Decompositions for Large Sparse Text Data Using Clustering},
  journal = {Machine Learning},
  volume  = {42},
  number  = {1--2},
  pages   = {143--175},
  year    = {2001},
  doi     = {10.1023/A:1007612920971}
}

@article{LeichtNewman2008,
  author  = {Leicht, E. A. and Newman, M. E. J.},
  title   = {Community Structure in Directed Networks},
  journal = {Physical Review Letters},
  volume  = {100},
  number  = {11},
  pages   = {118703},
  year    = {2008},
  doi     = {10.1103/PhysRevLett.100.118703}
}

@article{StrehlGhosh2002,
  author  = {Strehl, Alexander and Ghosh, Joydeep},
  title   = {Cluster Ensembles---A Knowledge Reuse Framework for Combining Multiple Partitions},
  journal = {Journal of Machine Learning Research},
  volume  = {3},
  pages   = {583--617},
  year    = {2002}
}

@article{HubertArabie1985,
  author  = {Hubert, Lawrence and Arabie, Phipps},
  title   = {Comparing Partitions},
  journal = {Journal of Classification},
  volume  = {2},
  pages   = {193--218},
  year    = {1985},
  doi     = {10.1007/BF01908075}
}

@misc{McDaidGreeneHurley2011,
  author        = {McDaid, Aaron F. and Greene, Derek and Hurley, Neil},
  title         = {Normalized Mutual Information to Evaluate Overlapping Community Finding Algorithms},
  year          = {2011},
  archivePrefix = {arXiv},
  eprint        = {1110.2515},
  primaryClass  = {physics.soc-ph}
}

@article{LeskovecKleinbergFaloutsos2007,
  author  = {Jure Leskovec and Jon Kleinberg and Christos Faloutsos},
  title   = {Graph Evolution: Densification and Shrinking Diameters},
  journal = {ACM Transactions on Knowledge Discovery from Data},
  volume  = {1},
  number  = {1},
  year    = {2007}
}

@article{WangLiangJi2020,
  author  = {Wang, Zhe and Liang, Yingbin and Ji, Pengsheng},
  title   = {Spectral Algorithms for Community Detection in Directed Networks},
  journal = {Journal of Machine Learning Research},
  volume  = {21},
  number  = {153},
  pages   = {1--45},
  year    = {2020}
}

@article{KimShi2012,
  author  = {Kim, Sungmin and Shi, Tao},
  title   = {Scalable Spectral Algorithms for Community Detection in Directed Networks},
  journal = {arXiv preprint arXiv:1211.6807},
  year    = {2012}
}

@article{DangDoPhan2023,
  author  = {Dang, Tien Dat and Do, Duy Hieu and Phan, Thi Ha Duong},
  title   = {Community Detection in Directed Graphs Using Stationary Distribution and Hitting Times Methods},
  journal = {Social Network Analysis and Mining},
  volume  = {13},
  number  = {80},
  year    = {2023},
  doi     = {10.1007/s13278-023-01080-1}
}

@article{PonsLatapy2006,
  author  = {Pons, Pascal and Latapy, Matthieu},
  title   = {Computing Communities in Large Networks Using Random Walks},
  journal = {Journal of Graph Algorithms and Applications},
  volume  = {10},
  number  = {2},
  pages   = {191--218},
  year    = {2006}
}

@article{RosvallBergstrom2008,
  author  = {Rosvall, Martin and Bergstrom, Carl T.},
  title   = {Maps of Random Walks on Complex Networks Reveal Community Structure},
  journal = {Proceedings of the National Academy of Sciences},
  volume  = {105},
  number  = {4},
  pages   = {1118--1123},
  year    = {2008}
}

@inproceedings{YangLeskovec2014CoDA,
  author    = {Yang, Jaewon and McAuley, Julian and Leskovec, Jure},
  title     = {Detecting Cohesive and 2-Mode Communities in Directed and Undirected Networks},
  booktitle = {Proceedings of the Seventh ACM International Conference on Web Search and Data Mining},
  pages     = {323--332},
  year      = {2014},
  doi       = {10.1145/2556195.2556243}
}

@misc{SNAPEmailEuCore,
  title        = {Email-Eu-core network dataset},
  author       = {{Stanford Network Analysis Project}},
  howpublished = {\url{https://snap.stanford.edu/data/email-Eu-core.html}},
  note         = {Accessed 2026}
}

@misc{SNAPCollegeMsg,
  title        = {CollegeMsg temporal network dataset},
  author       = {{Stanford Network Analysis Project}},
  howpublished = {\url{https://snap.stanford.edu/data/CollegeMsg.html}},
  note         = {Accessed 2026}
}

@misc{SNAPWikiVote,
  title        = {Wikipedia vote network dataset},
  author       = {{Stanford Network Analysis Project}},
  howpublished = {\url{https://snap.stanford.edu/data/wiki-Vote.html}},
  note         = {Accessed 2026}
}

@misc{SNAPGnutella04,
  title        = {Gnutella Peer-to-Peer Network, August 4 2002},
  author       = {{Stanford Network Analysis Project}},
  howpublished = {\url{https://snap.stanford.edu/data/p2p-Gnutella04.html}},
  note         = {Accessed May 2026}
}

@misc{SNAPGnutella08,
  title        = {Gnutella Peer-to-Peer Network, August 8 2002},
  author       = {{Stanford Network Analysis Project}},
  howpublished = {\url{https://snap.stanford.edu/data/p2p-Gnutella08.html}},
  note         = {Accessed May 2026}
}

@inproceedings{YinBensonLeskovecGleich2017,
  author    = {Yin, Hao and Benson, Austin R. and Leskovec, Jure and Gleich, David F.},
  title     = {Local Higher-Order Graph Clustering},
  booktitle = {Proceedings of the 23rd ACM SIGKDD International Conference on Knowledge Discovery and Data Mining},
  pages     = {555--564},
  publisher = {ACM},
  year      = {2017},
  doi       = {10.1145/3097983.3098069}
}

@article{PanzarasaOpsahlCarley2009,
  title={Patterns and Dynamics of Users' Behavior and Interaction: Network Analysis of an Online Community},
  author={Panzarasa, Pietro and Opsahl, Tore and Carley, Kathleen M.},
  journal={Journal of the American Society for Information Science and Technology},
  volume={60},
  number={5},
  pages={911--932},
  year={2009}
}

@inproceedings{KumarSpezzanoSubrahmanianFaloutsos2016,
  title={Edge Weight Prediction in Weighted Signed Networks},
  author={Kumar, Srijan and Spezzano, Francesca and Subrahmanian, V. S. and Faloutsos, Christos},
  booktitle={IEEE International Conference on Data Mining},
  pages={221--230},
  year={2016}
}

@article{WestPaskovLeskovecPotts2014,
  title={Exploiting Social Network Structure for Person-to-Person Sentiment Analysis},
  author={West, Robert and Paskov, Hristo S. and Leskovec, Jure and Potts, Christopher},
  journal={Transactions of the Association for Computational Linguistics},
  volume={2},
  pages={297--310},
  year={2014}
}

@article{RipeanuFosterIamnitchi2002,
  author  = {Ripeanu, Matei and Foster, Ian and Iamnitchi, Adriana},
  title   = {Mapping the {Gnutella} Network: Properties of Large-Scale Peer-to-Peer Systems and Implications for System Design},
  journal = {IEEE Internet Computing},
  volume  = {6},
  number  = {1},
  pages   = {50--57},
  year    = {2002}
}

@article{KarrerNewman2011,
  author  = {Karrer, Brian and Newman, M. E. J.},
  title   = {Stochastic blockmodels and community structure in networks},
  journal = {Physical Review E},
  volume  = {83},
  number  = {1},
  pages   = {016107},
  year    = {2011}
}

@inproceedings{QinRohe2013,
  author    = {Qin, Tai and Rohe, Karl},
  title     = {Regularized spectral clustering under the degree-corrected stochastic blockmodel},
  booktitle = {Advances in Neural Information Processing Systems},
  volume    = {26},
  year      = {2013}
}

@article{DoPhan2025OverlapCosine,
  author  = {Do, Duy Hieu and Phan, Thi Ha Duong},
  title   = {Overlapping Community Detection Algorithms Using Modularity and the Cosine},
  journal = {Advances in Complex Systems},
  volume  = {28},
  number  = {03},
  pages   = {2550006},
  year    = {2025},
  doi     = {10.1142/S0219525925500067}
}

@article{Hornik2012,
  author  = {Kurt Hornik and Ingo Feinerer and Martin Kober and Christian Buchta},
  title   = {Spherical k-Means Clustering},
  journal = {Journal of Statistical Software},
  volume  = {50},
  number  = {10},
  pages   = {1--22},
  year    = {2012}
}

@inproceedings{Perozzi2014,
  author    = {Bryan Perozzi and Rami Al-Rfou and Steven Skiena},
  title     = {DeepWalk: Online Learning of Social Representations},
  booktitle = {Proceedings of the 20th ACM SIGKDD International Conference on Knowledge Discovery and Data Mining},
  pages     = {701--710},
  year      = {2014}
}

@inproceedings{Tang2015,
  author    = {Jian Tang and Meng Qu and Mingzhe Wang and Ming Zhang and Jun Yan and Qiaozhu Mei},
  title     = {{LINE}: Large-scale Information Network Embedding},
  booktitle = {Proceedings of the 24th International Conference on World Wide Web},
  pages     = {1067--1077},
  year      = {2015}
}

@inproceedings{GroverLeskovec2016,
  author    = {Aditya Grover and Jure Leskovec},
  title     = {node2vec: Scalable Feature Learning for Networks},
  booktitle = {Proceedings of the 22nd ACM SIGKDD International Conference on Knowledge Discovery and Data Mining},
  pages     = {855--864},
  year      = {2016}
}

@article{LuZhou2011,
  author  = {Linyuan L{\"u} and Tao Zhou},
  title   = {Link prediction in complex networks: A survey},
  journal = {Physica A: Statistical Mechanics and its Applications},
  volume  = {390},
  number  = {6},
  pages   = {1150--1170},
  year    = {2011}
}

@article{Tandon2021,
  author  = {Tandon, Aditya and Albeshri, Aiiad and Thayananthan, Vijey and Alhalabi, Wadee and Radicchi, Filippo and Fortunato, Santo},
  title   = {Community Detection in Networks Using Graph Embeddings},
  journal = {Physical Review E},
  volume  = {103},
  number  = {2},
  pages   = {022316},
  year    = {2021},
  doi     = {10.1103/PhysRevE.103.022316}
}

@article{YuJiaoDehmerEmmertStreib2024,
  author  = {Yu, Guihai and Jiao, Yang and Dehmer, Matthias and Emmert-Streib, Frank},
  title   = {Community Detection in Directed Networks Based on Network Embeddings},
  journal = {Chaos, Solitons \& Fractals},
  volume  = {189},
  pages   = {115630},
  year    = {2024},
  doi     = {10.1016/j.chaos.2024.115630}
}

@article{DoPhan2025LouvainRW,
  author  = {Do, Duy Hieu and Phan, Thi Ha Duong},
  title   = {An Improvement on the {Louvain} Algorithm Using Random Walks},
  journal = {Journal of Combinatorial Optimization},
  volume  = {50},
  number  = {2},
  pages   = {14},
  year    = {2025},
  doi     = {10.1007/s10878-025-01337-9}
}

@inproceedings{DoNguyenPhan2025DFLouvainRW,
  author    = {Do, Duy Hieu and Nguyen, Dung and Phan, Thi Ha Duong},
  title     = {Improving the {DF-Louvain} Algorithm Through Random Walk-Based Refinement},
  booktitle = {Proceedings of the 2025 RIVF International Conference on Computing and Communication Technologies (RIVF)},
  pages     = {932--937},
  year      = {2025}
}

@inproceedings{DoPhan2022ExtendedWalktrap,
  author    = {Do, Duy Hieu and Phan, Thi Ha Duong},
  title     = {Detecting Communities in Large Networks Using the Extended {Walktrap} Algorithm},
  booktitle = {Proceedings of the 2022 RIVF International Conference on Computing and Communication Technologies (RIVF)},
  pages     = {100--105},
  year      = {2022},
  doi       = {10.1109/RIVF55975.2022.10013880}
}

\end{document}